\documentclass[reqno]{amsart}
\usepackage{amsmath,amssymb,cite}
\usepackage[mathscr]{euscript}

\theoremstyle{plain}
\newtheorem{theorem}{Theorem}[section]

\newtheorem{lemma}[theorem]{Lemma}

\theoremstyle{definition}

\theoremstyle{remark}
\newtheorem{remark}[theorem]{Remark}

\numberwithin{equation}{section}
\numberwithin{theorem}{section}

\def\be{\begin{equation}}
\def\ee{\end{equation}}
\def\bp{\begin{pmatrix}}
\def\ep{\end{pmatrix}}
\def\bea{\begin{eqnarray}}
\def\eea{\end{eqnarray}}

\def\\{\par\medskip}

\let\0=\noindent

\def\const{({\rm const.})\,}

\newcommand{\mc}[1]{{\mathcal #1}}

\newcommand{\bb}[1]{{\mathbb #1}}

\newcommand{\rme}{\mathrm{e}}

\newcommand{\rmd}{\mathrm{d}}
\newcommand{\eps}{\varepsilon}

\newcommand{\supp}{\mathop{\rm supp}\nolimits}

\title[Long time evolution of concentrated Euler flows]{Long time evolution of concentrated Euler flows with planar symmetry}

\author[P.\ Butt\`a]{Paolo Butt\`a}
\address{Paolo Butt\`a\hfill\break \indent
   Dipartimento di Matematica, 
   Sapienza Universit\`a di Roma,
      \hfill\break \indent
   P.le Aldo Moro 5, 00185 Roma, Italy}
 \email{butta@mat.uniroma1.it}
\author[C.\ Marchioro]{Carlo Marchioro}
\address{Carlo Marchioro \hfill\break \indent
   Dipartimento di Matematica, 
   Sapienza Universit\`a di Roma,
   \hfill\break \indent
   P.le Aldo Moro 5, 00185 Roma, Italy \hfill\break \indent and \hfill\break \indent  International Research Center M\&MOCS, Universit\`a di L'Aquila \hfill\break \indent Palazzo Caetani, 04012 Cisterna di Latina (LT), Italy}
\email{marchior@mat.uniroma1.it}

\begin{document}

\begin{abstract}
We study the time evolution of an incompressible Euler fluid with planar symmetry when the vorticity is initially concentrated in small disks. We discuss how long this concentration persists, showing that in some cases this happens for quite long times. Moreover, we analyze a toy model that shows a similar behavior and gives some hints on the original problem.
\end{abstract}
\keywords{Incompressible Euler flow, point vortex model, long time behavior}

\subjclass[2010]{
76B47,  
37N10,  
70K65.  
}

\maketitle
\thispagestyle{empty}

\section{Introduction}
\label{sec:1}

This paper focuses on the long time behavior of an incompressible inviscid fluid, with  planar symmetry and constant density, whose time evolution is governed by the two-dimensional Euler equations. If the system is confined in a domain $\Gamma\subseteq \bb R^2$, the Euler equations expressed in term of the vorticity read,
\begin{eqnarray}
\label{Eq.1}
&& \partial_t\omega (x,t) + (u\cdot\nabla)\omega(x,t) = 0\;, \quad x=(x_1,x_2)\in \Gamma\;, 
\\ \label{Eq.2}
&& \nabla \cdot u(x,t)=0\;, 
\\ \label{Eq.3}
&& \omega(x,0) = \omega_0(x)\;,
\end{eqnarray}
where $\omega := \partial_1 u_2-\partial_2 u_1 $ is the vorticity and $u=(u_1,u_2)$ denotes the velocity field. By assuming that $u$ vanishes at infinity and that its normal component vanishes on $\partial\Gamma$, the velocity is reconstructed from the vorticity as
\begin{equation}
\label{vel-vor1}
u(x,t) = \int\! \rmd y\, K_\Gamma(x,y)\,\omega(y,t)\;,
\end{equation}
with
\begin{equation}
\label{vel-vor2}
K_\Gamma = \nabla^\perp  G_\Gamma\;, \quad \nabla^\perp = (\partial_2,-\partial_1)\;,
\end{equation}
where $G_\Gamma$ is the fundamental solution of the Laplace operator in $\Gamma$ vanishing on the boundary (and at infinity if $\Gamma$ is unbounded). In particular, if $\Gamma=\bb R^2$, 
\begin{equation}
\label{vel-vor3}
K_{\bb R^2}(x,y) = K(x-y) = \nabla^\perp  G(x-y)\;, \quad G(x) = - \frac{1}{2\pi} \log|x|\;.
\end{equation}

Equation \eqref{Eq.1} means that the vorticity remains constant along the particle paths, which are the characteristics of the Euler equations. Otherwise stated,
\begin{equation}
\label{Cons1}
\omega(x(x_0,t),t) = \omega_0(x_0) \ ,
\end{equation}
where  $x(x_0,t)$ is the trajectory of the fluid particle initially in $x_0$, i.e., 
\begin{equation}
\label{Cons2}
\frac{\rmd}{\rmd t}x(x_0,t) = u(x(x_0,t),t)\;, \quad x(x_0,0) = x_0\;.
\end{equation}

It is possible to consider non smooth initial data, by assuming directly \eqref{vel-vor1}, \eqref{Cons1}, and \eqref{Cons2} as a weak formulation of the Euler equations. Indeed, see, e.g., \cite{MaP94}, given $\omega_0\in L^1(\Gamma)\cap L^\infty (\Gamma)$ and $T>0$, there exists a unique triple $(x(\cdot,\cdot),u,\omega)$ solution to \eqref{vel-vor1}, \eqref{Cons1}, and \eqref{Cons2} with $\omega \in L^\infty([0,T]; L^1\cap L^\infty)$. Moreover, since $u$ is divergence-free, the time evolution preserves the Lebesgue measure in $\Gamma$. In particular, given any smooth function $f(x,t)$ with compact support in $\Gamma\times [0,T]$,  if
\begin{equation}
\label{Aver}
\omega_t[f] := \int_\Gamma\! \rmd x\,\omega (x,t) f(x,t) = \int_\Gamma\! \rmd x_0\,\omega_0 (x) f(x(x_0,t),t) 
\end{equation}
then $t\mapsto \omega_t[f]$ belongs to $C^1([0,T])$ and 
\begin{equation}
\label{Weakeq}
\frac{\rmd}{\rmd t}\omega_t[f] = \omega_t [u \cdot \nabla f] + \omega_t[\partial_t f]\;.
\end{equation}
We remark that if $\omega_0$ has compact support then \eqref{Aver} and \eqref{Weakeq} are valid for any smooth function $f(x,t)$ (also with noncompact support).

In this paper we consider initial data in which the vorticity is supported in $N$ blobs, i.e., initial data of the form,
\begin{equation}
\label{in}
\omega_\eps(x,0) = \sum_{i=1}^N  \omega_{i,\eps}(x,0)\;,
\end{equation}
where $\omega_{i,\eps}(x,0)$, $i=1,\ldots, N$, are functions with definite sign such that, denoting by $\Sigma(z|r)$ the open disk of center $z$ and radius $r$, 
\begin{equation}
\label{initial}
\Lambda_{i,\eps}(0) := \supp\, \omega_{i,\eps}(\cdot,0) \subset \Sigma(z_i|\eps)\;, \quad  \Sigma(z_i|\eps)\cap \Sigma(z_j|\eps)=\emptyset\quad \forall\, i \ne j\;, 
\end{equation}
with $\eps \in (0,1)$ a small parameter and the points $z_i\in \Gamma$ such that the closure of $\Sigma(z_i|\eps)$ does not intersect the boundary of $\Gamma$ for any $i=1,\ldots,N$. In general, the signs of the functions $\omega_{i,\eps}(x,0)$ can be different among each other.

As is well known in the literature, for such initial data the dynamics can be approximated by the following system of $N$ differential equations in $\Gamma$, known as the \textit{point vortex model},
\begin{equation}
\label{Pointv}
\dot{z}_i(t) = \sum_{\substack{j=1 \\ j\ne i}}^N a_j K_\Gamma(z_i(t),z_j(t)) + \frac 12 a_i \nabla^\perp \gamma_\Gamma(z_i(t)) \;, \quad z_i(0)=z_i\;, 
\end{equation}
where
\begin{equation}
\label{point vi}
a_i = \int_\Gamma\! \rmd x \, \omega_{i,\eps}(x,0)
\end{equation}
is called the ``intensity'' of the vortex and it is assumed independent of $\eps$, while  $\gamma_\Gamma(x) = \gamma_\Gamma(x,x)$ with $\gamma_\Gamma(x,y) :=  G_\Gamma(x,y) - G(x,y)$, the ``regular part'' of the Green function $G_\Gamma$. 

In particular, it has been proved \cite {MaP93,MaP94,Mar98,CaM} that the time evolution of these states has, for small $\eps$, a similar form,
\begin{equation}
\label{evolved}
\omega_\eps(x,t) = \sum_{i=1}^N  \omega_{i,\eps}(x,t)\;, 
\end{equation}
where $\omega_{i,\eps}(x,t)$, $i=1,\ldots,N$, are functions with definite sign such that
\begin{equation}
\label{evolvedbis}
\begin{split}
& \Lambda_{i,\eps}(t) := \supp\, \omega_{i,\eps}(\cdot,t) \subset \Sigma(z_i(t) | r_t(\eps))\;, \\ & \Sigma(z_i(t) | r_t(\eps))\cap \Sigma(z_j(t) | r_t(\eps)) = \emptyset \quad \forall\, i \ne j\;, 
\end{split}
\end{equation}
with $\{z_i(t)\}_{i=1}^N$ satisfying \eqref{Pointv} and $r_t(\eps)$ a nonnegative function such that the closure of $\Sigma(z_i|r_t(\eps))$ does not intersect the boundary of $\Gamma$ for any $i=1,\ldots,N$.

The point vortex model \eqref{Pointv} was introduced in the eighteenth century by Helmholtz \cite{Hel}, as a particular ``solution'' of the Euler equations, and investigated by several authors \cite {Kir,Poi,Kel}, see \cite{MaP94} and references quoted therein for a review on this subject. This model approximates reasonably a state with very large vorticity concentrations.

In general, the point vortex model admits a global solution, but in some cases collapses can happen \cite {Are07}. However, it can be shown that the set of initial data and vortex intensities that produce a blow-up  is exceptional; see, respectively, \cite {DuP,MaP84,MaP94} for the proof in the case of the torus, the disk, and the plane. Moreover, there are initial data for which the vortices move away from each other indefinitely (we will return to this point later on).  

In any case, for each time $t$ chosen before a possible collapse, it can be proved that $r_t(\eps) \to 0$ for $\eps\to 0$ and the fluid converges to the point vortex system \cite{MaP93,MaP94,Mar98,CaM}. For the connection between the Euler flow and the point vortices, see also \cite {Gal,Mar88,Mar90,Ma98,MaP,Mar99,Mar07,MaP83,Tur}.
 
Suppose now that the initial datum \eqref{in}-\eqref{initial} is chosen in such a way that the corresponding Cauchy problem \eqref{Pointv} admits a global solution. For what stated before, the fluid converges to the point vortex system for any fixed positive time. On the other hand, in a realistic situation the parameter $\eps$ is not zero, hence a natural question is to characterize the larger time intervals on which this  approximation is valid for small but positive values of the parameter $\eps$. 

As time goes by, small filaments of fluid could move away. We fix $\beta\in (0,1/2)$ (we will see the technical reason for this limitation) and we denote by $T_{\eps,\beta}$ the first time in which a filament reaches the boundary  of $\bigcup_i \Sigma(z_i(t)|\eps^\beta)$. Clearly, $T_{\eps,\beta}$ gives a lower bound of the time horizon where the point vortex approximation is valid. In the general case, by adapting the strategies given \cite{MaP93,MaP94,Mar98,CaM}, we can show that $T_{\eps,\beta} \ge \const  |\log \eps |$ for small $\eps$. 

This bound is poor, and perhaps the result is too naive. Let us discuss this point. When $\Gamma=\bb R^2$ and there is a vortex alone, the center of the vorticity blob remains fixed and the spread of vorticity grows in time slowly, see \cite{Mar94,Mar96,LoN,Ser,ISG}, where it is shown that there is $c>0$ such that $T_{\eps,\beta} \ge \eps^{-c}$ for $\eps$ small, that is $T_{\eps,\beta}$ admits a power-law lower bound (on this point see Section \ref{sec:3}). The presence of the interaction with other blobs of vorticity and/or with the boundary of $\Gamma$ produces a priori a larger spread. A possible bound is $T_{\eps,\beta} \ge \const  |\log \eps |$, but we conjecture that it could be improved. In some particular cases, a power-law upper bound for $T_{\eps,\beta}$ can be obtained rigorously by a direct analysis of the problem. In more general cases, it could be obtained by making an average on time of the interaction with the other blobs. Indeed, when $\eps$ is small the fluid in a blob turns very quickly and so the effects of the other blobs depend on the time average of the interaction. This problem appears very challenging and it is not rigorously analyzed. To have some hints in this direction, we introduce a very schematic toy model and we investigate, by using a second order averaging method, the long time behavior of its solutions.

Another mechanism that should improve the convergence of the Euler flow to the point vortex model could be a very careful preparation of the initial data. Actually, in some textbooks in fluid mechanics (see for instance \cite {Bat}), the spread of a blob of vorticity due to the interaction with other blobs is neglected for symmetry reasons, assuming the initial data with a radial symmetry. Of course, the time evolution destroys this symmetry, but we can hope that until some time this property remains almost valid. We will show (at the end of Section 4) that the analysis of the aforementioned toy model allows us to strengthen this conjecture. 

We end this introduction observing that there are some sequences (in $\eps$) that imply better estimates on $T_{\eps,\beta}$. A trivial example is given by a vortex alone in the plane with a vorticity with radial symmetry: of course, this is a stationary state of  the Euler equations and so $T_{\eps,\beta}=\infty$. We could look for different and less trivial situations, with many blobs of vorticity that are stationary in time, but here we are interested in not so exceptional cases.

Actually, the relation between special dynamical systems (like the point vortex model) and the fluid mechanics is always related to some a priori assumptions. In the present case, we study situations with planar symmetry. In other cases, we assume other symmetries, often renouncing to study only compact blobs of vorticity. For instance, let us consider cylindrical symmetry without swirl: using cylindrical coordinates $(z,r,\theta)$, the motion does not depend on $\theta$ and it can be described in the $(z,r)$ plane, where a point represents a ring in the whole space. An approximated ring converges as $\eps \to 0$ to a ring that performs a rigid translation in the $z$-direction \cite{BCM00} ($\eps$ is the size of the tube of vorticity around the ring). In this case, we must renormalize the total vorticity by a factor $|\log \eps |^{-1}$. The same problem has been studied for $r \approx |\log \eps|$. With this choice, the tubes of vorticity are expected to converge to rings whose evolution is governed by a dynamical system similar but not equal to the point vortex system. This has been rigorously proved in the case of one ring alone in \cite{MaN99}, while the case of many rings remains an open problem.

For other examples of dynamical systems related to fluid mechanics, see for instance \cite {KMD,MB} and the references quoted in the recent paper \cite {CGC}. Their connection with the fluid physics, proved in some particular cases, is in general an open issue. 

We conclude with a final remark. Here, we discuss the fluid mechanics with planar symmetry, i.e., when a point vortex in $\bb R^2$ represents an infinite straight line in $\bb R^3$. We could also study the aforementioned case with cylindrical symmetry without swirl, in which the straight line becomes a circle of radius $r$, and consider the case of $N$ blobs of vorticity in the plane $(z,r)$ of size $\eps$ and centered around the points $(z_i,r_i)$. Let us make the change of variable $z=x, r=r_0+y$. We increase $r_0$ as $\eps$ decreases choosing  $r_0=\eps^{-b}, b>0$. It has been proved in \cite {Mar99} that in the limit $\eps \to 0$ the Euler flow converges to the point vortex system (\ref{Pointv}). Hence, we could apply also in this case our investigation.

The plan of the paper is the following. In the next section we discuss how to obtain, in general, the logarithmic  lower bound on $T_{\eps,\beta}$. In Section \ref{sec:3}, we give examples, in the whole plane and in a disk, where a power-law lower bound on $T_{\eps,\beta}$ holds true. In Section \ref{sec:4}, we introduce the toy model whose long-time behavior suggests similar features of the fluid dynamics.

\section{Persistence of vortices on logarithmic time scales}
\label{sec:2}

In this section we consider the general case of initial data of the form \eqref{in}, \eqref{initial}, with the only requirement that the associated Cauchy problem \eqref{Pointv} of the point vortex dynamics admits a global solution such that 
\begin{equation}
\label{rmin}
r_\mathrm{min} := \inf_{t\ge 0} \min_{i\ne j} \;\min\{|z_i(t)-z_j(t)|;\mathrm{dist}(z_i(t);\partial\Gamma)\} > 0\;.
\end{equation} 

To each initial data \eqref{in}, \eqref{initial} satisfying the above assumption and $\beta>0$ we associate the variable $T_{\eps,\beta}$ defined in the Introduction, i.e.,
\begin{equation}
\label{tbe}
T_{\eps,\beta} := \min_i \sup\{t>0 \colon |x(x_0,s)-z_i(s)| < \eps^\beta\; \;\forall\, s\in [0,t]\;\; \forall\, x_0\in \Lambda_{i,\eps}(0)\}\;.
\end{equation}
Our task is a lower bound on $T_{\eps,\beta}$. For simplicity, we analyze here the case $\Gamma =\bb R^2$, but the proof can be easily adapted to the case of a general domain. 

We recall that for vortices with intensities of the same sign, \eqref{rmin} is a well known property of the dynamics. In the general case, the existence of a unique global solution to the Cauchy problem \eqref{Pointv} is proved for any choice of initial data and intensities $\{z_i,a_i\}_{i=1}^N$, outside a set of Lebesgue measure zero \cite{MaP94}. This fact does not implies \eqref{rmin}, but it makes this assumption very reasonable.

\begin{theorem} 
\label{thm:1}
Let $\Gamma=\bb R^2$ and assume that the initial data of the Euler equations verify the above assumptions. Suppose also that there are $M,\nu>0$ such that
\begin{equation}
\label{bound}
|\omega_{i,\eps}(x,0)| \le M \eps^{-\nu}\;.
\end{equation}
Then, for each $\beta\in (0,1/2)$ there exist $\eps_0>0$ and $\zeta_0>0$ such that
\begin{equation}
\label{tbeb1}
T_{\eps,\beta} > \zeta_0 |\log \eps | \quad\forall\,\eps\in(0,\eps_0)\;. 
\end{equation}
\end{theorem}

We split the proof into two steps. First, in the next subsection, we prove an analogous result for a reduced system: the motion of a single blob of vorticity in an external time-dependent divergence-free vector field. The original problem is then solved by using the reduced system to simulate the force acting on a given blob of vorticity due to its interaction with the other blobs.

\subsection{The reduced system}
\label{sec:2.1}
We consider a single blob of vorticity which evolves in $\bb R^2$ in presence of an external time-dependent divergence-free vector field $F(x,t)$. This means that (i) the initial configuration $\omega_\eps(x,0)$ is a function of definite sign such that $\Lambda_\eps(0):=\supp\omega_\eps(\cdot,0)\subset \Sigma(z_*|\eps)$ for some $z_*\in\bb R^2$ and (ii) the evolved configuration $\omega(x,t) = \omega_\eps(x,t)$ satisfies \eqref{Cons1}, and with in this case $x(x_0,t)$ solution to
\begin{equation}
\label{Cons4}
\frac{\rmd}{\rmd t}x(x_0,t) = u(x(x_0,t),t)+F(x(x_0,t),t)\;, \quad x(x_0,0)=x_0\;,
\end{equation}
where $u(x,t) = \int\! \rmd y\, K(x-y)\,\omega_\eps(y,t)$ with $K$ as in \eqref{vel-vor3}. As a consequence, the weak formulation \eqref{Weakeq} is replaced by  
\begin{equation}
\label{WeakeqF}
\frac{\rmd}{\rmd t}\omega_t[f] = \omega_t[(u+F) \cdot \nabla f]+\omega_t[\partial_t f]\;.
\end{equation}

Since the auxiliary field $F(x,t)$ will be used to simulate the action of the other blobs of vorticity, we can assume that it is bounded and, with respect to the spatial variable, divergence-free and Lipschitz, 
\begin{equation}
\label{2Lipsc}
\|F\|_\infty < +\infty\;, \quad |F(x,t)-F(y,t)| \le  D_t |x-y|\;, \quad D:=\sup_{t\in [0,+\infty)} D_t < +\infty\;. 
\end{equation}

The point vortex dynamics associated to the reduced system is defined by the planar motion $B(t)$,  solution to the following equation,
\begin{equation}
\label{dis}
\dot B(t) = F(B(t),t)\;, \quad  B(0)=z_*\;.
\end{equation} 
Without loss of generality, we also assume that initially, and hence at any time, the blob has intensity one, 
\begin{equation}
\label{w=1}
\omega_\eps(x,t)\ge 0\;, \quad \int\!\rmd x\, \omega_\eps(x,t) = 1\;.
\end{equation}

For this reduced system we prove the following result.
\begin{theorem} 
\label{thm:2}
Let $\Lambda_\eps(t):=\supp\omega_\eps(\cdot,t)$, suppose there are $M,\nu>0$ such that
\begin{equation}
\label{bound1}
\omega_\eps(x,0) \le M \eps^{-\nu}\;,
\end{equation}
and define
\begin{equation}
\label{tbe*}
T_{\eps,\beta}^* := \sup\{t>0 \colon \Lambda_\eps(s) \subset \Sigma(B(s)|\eps^\beta) \;\; \forall\, s\in [0,t]\}\;.
\end{equation}
Then, for each $\beta\in (0,1/2)$ there exist $\eps_1>0$ and $\zeta_1>0$ such that
\begin{equation}
\label{tbeb2}
T_{\eps,\beta}^* > \zeta_1 |\log \eps | \quad\forall\,\eps\in(0,\eps_1)\;. 
\end{equation}
\end{theorem}

The proof is similar to that in \cite{CaM} and it is based on a bootstrap argument. For later purposes, it is useful to separate the principal estimates in different lemmas, giving the proof of the theorem at the end of the subsection. 

We denote by $B_\eps(t)$ the center of vorticity of the blob, defined by
\begin{equation}
\label{c.m.}
B_\eps(t) = \int\! \rmd x\, x\, \omega_\eps(x,t)\;, 
\end{equation}
and by $I_\eps(t)$ the moment of inertia with respect of $B_\eps$, i.e.,
\begin{equation}
\label{moment}
I_\eps(t) = \int\! \rmd x\, |x-B_\eps(t)|^2  \omega_\eps(x,t)\;.
\end{equation}

\begin{lemma}
\label{lem:1}
For any $t\ge 0$, the following estimates hold, 
\begin{equation}
\label{Iee}
I_\eps(t) \le 4\eps^2 \exp\bigg[2\int_0^t\!\rmd s\, D_s\bigg]\;,
\end{equation}
\begin{equation}
\label{bee}
|B_\eps(t)-B(t)| \le 2\eps \bigg(1 + \int_0^t\!\rmd s\, D_s\bigg) \exp\bigg[\int_0^t\!\rmd s\, D_s\bigg]\;,
\end{equation}
where $D_t$ is the Lipschitz constant introduced in \eqref{2Lipsc}.
\end{lemma}

\begin{proof}
From \eqref{Cons1}, \eqref{Cons4}, and since $u+F$ is divergence-free we have,
\begin{equation*}
B_\eps(t) = \int\! \rmd x_0\, x(x_0,t) \, \omega_\eps(x_0,0)\;, \quad  I_\eps(t) = \int\! \rmd x_0\, |x(x_0,t)-B_\eps(t)|^2  \omega_\eps(x_0,0)\;.
\end{equation*}
Therefore, by \eqref{Cons4} and using the identities 
\begin{equation}
\label{id-K}
\int\!\rmd x\, u(x,t)\,\omega_\eps(x,t) = 0\;,\quad \int\!\rmd x\, x\cdot u(x,t)\,\omega_\eps(x,t) = 0,
\end{equation}
the time derivatives of $B_\eps(t)$ and $I_\eps(t)$ are easily computed, 
\begin{equation}
\label{growth B}
\dot B_\eps(t) =  \int\! \rmd x\, F(x,t) \,\omega_\eps(x,t)\;,
\end{equation}
\begin{equation}
\label{growth moment}
\dot I_\eps(t) = 2 \int\! \rmd x\, (x-B_\eps(t))\cdot F(x,t)\, \omega_\eps(x,t)\;.
\end{equation}

By \eqref{2Lipsc} and the obvious identity,
\begin{equation*}
\int\! \rmd x\, (x-B_\eps(t)) \cdot F(B_\eps(t),t) \, \omega_\eps(x,t) = 0\;,
\end{equation*}
we have, 
\begin{equation*}
|\dot I_\eps(t)| \le 2 D_t  \int\! \rmd x\,  |x-B_\eps(t)| ^2 \, \omega_\eps(x,t) = 2 D_t I_\eps(t)\;,
\end{equation*}
which can be integrated, giving
\begin{equation*}
I_\eps(t) \le I_\eps(0) \exp\bigg[2\int_0^t\!\rmd s\, D_s\bigg]\;,
\end{equation*}
which implies \eqref{Iee} because of the (not optimal) estimate $I_\eps(0) \le 4\eps^2$, following immediately from the fact that $\Lambda_\eps(0)\subset\Sigma(z_*|\eps)$ and in view of \eqref{w=1}.

To prove \eqref{bee}, we observe that, by \eqref{dis}, \eqref{growth B}, and \eqref{w=1},
\begin{equation*}
\dot B_\eps(t) - \dot B(t) =  F(B_\eps(t),t) - F(B(t),t) + \int\!\rmd x\, [F(x,t) - F(B_\eps(t),t)] \, \omega_\eps(x,t) \;.
\end{equation*}
Therefore, by \eqref{2Lipsc},
\begin{equation*}
\begin{split}
|\dot B_\eps(t) - \dot B(t)| & \le D_t |B_\eps(t)-B(t)| + D_t \int\!\rmd x\,  |B_\eps(t)- x |\,\omega_\eps(x,t) \\ & \le  D_t |B_\eps(t)-B(t)| +  D_t \sqrt{I_\eps(t)}\;,
\end{split}
\end{equation*}
where in the last estimate we used the Cauchy-Schwarz inequality and \eqref{w=1}. The last differential inequality can be integrated, getting
\begin{equation*}
|B_\eps(t)-B(t)|  \le |B_\eps(0)-z_*| \exp\bigg[\int_0^t\!\rmd s\, D_s\bigg] + \int_0^t\!\rmd s\, D_s\sqrt{I_\eps(s)}\, \exp\bigg[\int_s^t\!\rmd \tau\, D_\tau\bigg] \;,
\end{equation*}
which implies \eqref{bee} in view of \eqref{Iee} and since $|B_\eps(0)-z_*|\le 2\eps$.
\end{proof}

\begin{remark}
\label{rem:2.1}
In the proofs of Theorems \ref{thm:1} and \ref{thm:2}, the estimates of Lemmas \ref{lem:1} and \ref{lem:2} will be used with $D$ as in \eqref{2Lipsc} instead of $D_t$. Nevertheless, we keep the formulation involving the integral of $D_t$ because this will be used later in the proof of Theorem \ref{thm:3}. We also remark that the identities \eqref{id-K} follow from the antisymmetry of $K=K_{\bb R^2}$. We observe that these are no longer true for a general domain $\Gamma$. On the other hand, in this case $K_\Gamma = K+\tilde K$ with $\tilde K$ a kernel which is regular away from the boundary, so that its contribution can be treated as an external (bounded, divergence-free, and Lipschitz) field, see also the Remark \ref{rem:2.2} at the end of this section.  
\end{remark}

Now, we introduce a positive parameter $\alpha$, to be chosen small enough, and study the system for times $0 \le t \le \alpha |\log \eps |$. Recalling the definition of $D$ in \eqref{2Lipsc}, by \eqref{Iee} and \eqref{bee} we have,
\begin{eqnarray}
\label{2bound moment}
& I_\eps(t)  \le \ 4 \eps^\delta \quad \forall\, t \in [0, \alpha |\log \eps |]\;, \\ \label{2concl2} & |B_\eps(t)-B(t)|  \le 2(1+D\alpha |\log\eps|) \eps^{\delta/2} \quad \forall\, t \in [0, \alpha |\log \eps |]\;,
\end{eqnarray}
with $\delta=2-2D\alpha>0$, provided $\alpha$ is small enough. 

The bound \eqref{2bound moment} implies that for small $\eps$ the main part of the vorticity remains concentrated around $B_\eps(t)$, which, in turn, remains close to $B(t)$ in view of \eqref{2concl2}. We now prove that not only the main part but all the filaments of vorticity remain close to $B_\eps(t)$. 

To this purpose, we study the growth in time of the distance from $B_\eps(t)$ of a fluid particle. The key point is to show that this growth is very small for the particles sufficiently away from the center of vorticity. This is a consequence of the following two lemmas. 

\begin{lemma}
\label{lem:2}
Recall $\Lambda_\eps(t)=\supp\omega_\eps(\cdot,t)$ and define
\begin{equation}
\label{Rt}
R_t:= \max\{|x-B_\eps(t)|\colon x\in \Lambda_\eps(t)\}\;.
\end{equation}
Given $x_0\in\Lambda_\eps(0)$, suppose at time $t>0$ it happens that  
\begin{equation}
\label{hstimv}
|x(x_0,t)-B_\eps(t)| = R_t\;.
\end{equation}
Then, at this time $t$, 
\begin{equation}
\label{stimv}
\frac{\rmd}{\rmd t} |x(x_0,t)- B_\eps(t)| \le 2D_tR_t + \frac{5I_\eps(t)}{\pi R_t^3} + \sqrt{\frac{M\eps^{-\nu}\,m_t(R_t/2)}{\pi}}\;,
\end{equation}
with $M,\nu$ as in \eqref{bound1} and the function $m_t(\cdot)$ on $\bb R_+$ defined by 
\begin{equation}
\label{mt}
m_t(h) = \int_{|y-B_\eps(t)|>h}\!\rmd y\;\omega_\eps(y,t)\;.
\end{equation}
\end{lemma}

\begin{proof}
We observe that part of the proof is similar to that given in \cite{ISG} (with $B_\eps(t) = 0$). Letting $x=x(x_0,t)$, by \eqref{Cons4}, \eqref{growth B}, and \eqref{w=1} we have,
\begin{equation}
\label{distance1}
\begin{split}
& \frac{\rmd}{\rmd t} |x(x_0,t)- B_\eps(t)| = \big(u(x,t) + F(x,t) - \dot B_\eps(t)\big) \cdot \frac{x-B_\eps(t)}{|x-B_\eps(t)|} \\ & \;\; = \bigg[\int\!\rmd y\, \big(F(x,t) - F(y,t) +K(x-y)\big)\, \omega_\eps(y,t)\bigg] \cdot  \frac{x-B_\eps(t)}{|x-B_\eps(t)|}\;.
\end{split}
\end{equation}
The first term in the second line, due the external field, is easily bounded by using \eqref{2Lipsc}, \eqref{w=1}, and \eqref{Rt},
\begin{equation}
\label{distance4}
\bigg|\int\!\rmd y\, [F(x,t) - F(y,t)]\, \omega_\eps(y,t)\bigg| \le D_t\int\!\rmd y\, |x-y| \, \omega_\eps(y,t) \le 2D_t R_t\;.
\end{equation}
For the second term, we split the integration region into two parts, the disk $A_1=\Sigma(B_\eps(t)|R_t/2)$ and the annulus $A_2 = \Sigma(B_\eps(t)|R_t)\setminus\Sigma(B_\eps(t)|R_t/2)$. Then,
\begin{equation}
\label{in A_1,A_2}
\frac{x-B_\eps(t)}{|x-B_\eps(t)|} \cdot \int\! \rmd y\, K(x-y)\, \omega_\eps(y,t) = H_1 + H_2\;,
\end{equation}
where
\begin{equation}
\label{in A_1}
H_1 = \frac{x-B_\eps(t)}{|x-B_\eps(t)|} \cdot \int_{A_1}\! \rmd y\, K(x-y)\, \omega_\eps(y,t) 
\end{equation}
and
\begin{equation}
\label{in A_2}
H_2 = \frac{x-B_\eps(t)}{|x-B_\eps(t)|} \cdot \int_{A_2}\! \rmd y\, K(x-y)\, \omega_\eps(y,t)\;.
\end{equation}

We first evaluate the contribution of the integration on $A_1$. Recalling \eqref{vel-vor3} and the notation $x^\perp = (x_2,-x_1)$ for $x = (x_1,x_2)$, after introducing the new variables $x'=x-B_\eps(t)$,  $y'=y-B_\eps(t)$, and using that $x'\cdot (x'-y')^\perp=-x'\cdot y'^\perp$, we get,
\begin{equation}
\label{in H_11}
H_1 = \frac{1}{2\pi} \int_{|y'|\leq R_t/2}\! \rmd y'\, \frac{x'\cdot y'^\perp}{|x'||x'-y'|^2}\, \omega_\eps(y'+B_\eps(t))\;.
\end{equation}
By \eqref{c.m.}, $\int\! \rmd y'\,  y'^\perp\, \omega_\eps(y'+B_\eps(t)) = 0$, so that
\begin{equation}
\label{in H_13}
H_1  = H_1'-H_1''\;, 
\end{equation}
where
\begin{eqnarray*}
&& H_1' = \frac{1}{2\pi}  \int_{|y'|\le R_t/2}\! \rmd y'\, \frac {x'\cdot y'^\perp}{|x'|}\, \frac {y'\cdot (2x'-y')}{|x'-y'|^2 \ |x'|^2} \, \omega_\eps(y'+B_\eps(t))\;, \\ && H_1''= \frac{1}{2\pi} \int_{|y'|> R_t/2}\! \rmd y'\, \frac{x'\cdot y'^\perp}{|x'|^3}\, \omega_\eps(y'+B_\eps(t))\;.
\end{eqnarray*}
From \eqref{hstimv} we have $|x'| = R_t$, and hence $|y'| \le R_t/2$ implies $|x'-y'|\ge R_t/2$ and $|2x'-y'|\le |x'-y'|+|x'| \le 3 |x'-y'|$, so that
\begin{equation*}
|H_1'|\leq \frac{3}{\pi R_t^3}  \int_{|y'|\leq R_t/2} \! \rmd y'\, |y'|^2 \, \omega_\eps(y'+B_\eps(t)) \le \frac{3 I_\eps(t)}{\pi R_t^3}\;.
\end{equation*}
To bound $H_1''$ we note that, in view of \eqref{Rt}, the integration is restricted to $|y'|\le R_t$, so that
\begin{equation*}
|H_1''| \le \frac{1}{2\pi R_t} \int_{|y'|> R_t/2}\! \rmd y'\, \omega_\eps(y'+B_\eps(t))\le \frac{2I_\eps(t)}{\pi R_t^3}\;,
\end{equation*}
where in the last inequality we used the Chebyshev's inequality. By \eqref{in H_13} and the previous estimates we conclude that
\begin{equation}
\label{H_14}
|H_1| \le \frac{5 I_\eps(t)}{\pi R_t^3}\;.
\end{equation}

We now evaluate $H_2$. Recalling the definition of $K$, 
\begin{equation*}
|H_2| \le \frac{1}{2\pi} \int_{A_2}\! \rmd y\, \frac 1{|x-y|} \, \omega_\eps(y,t)\;.
\end{equation*}
The integrand is monotonically unbounded as $y\to x$, and so the maximum of the integral is obtained when we rearrange the vorticity mass as close as possible to the singularity. In view of the assumption \eqref{bound1} and since, by \eqref{mt}, $m_t(R_t/2)$ is equal  to the total amount of vorticity in $A_2$, this rearrangement gives,\footnote{Here, we estimate \begin{equation*}\int_{A_2}\! \rmd y\, \frac{\omega_\eps(y,t)}{|x-y|} \le \max\bigg\{\int\! \rmd y'\, \frac {\omega(y')}{|y'|} \colon \int\!\rmd y'\,\omega(y') = m_t(R_t/2)\;, \; 0\le \omega(\cdot) \le M\eps^{-\nu}\bigg\}\;,\end{equation*} and we explicitly find the distribution of vorticity that realizes this maximum. By rearrangement, this is the piecewise constant function given by  the maximum value $M\eps^{-\nu}$ on the disk $\Sigma(0|r)$ and zero otherwise, with $r$ such that the mass constraint is satisfied. Alternatively, we could use \cite[Lemma 2.1]{ISG}, getting a little bit worst estimate for the integral in the left-hand side, giving rise to a constant greater than one in front of the third term in the right-hand side of \eqref{stimv}.}
\begin{equation}
\label{h2}
|H_2| \le \frac{M\eps^{-\nu}}{2\pi} \int_{\Sigma (0|r)}\!\rmd y'\, \frac{1}{|y'|} = M\eps^{-\nu} r \;, 
\end{equation}
where the radius $r$ is such that $M \eps^{-\nu} \pi r^2 = m_t(R_t/2)$. The estimate \eqref{stimv} now follows by \eqref{distance1}, \eqref{distance4}, \eqref{in A_1,A_2}, \eqref{H_14}, and \eqref{h2}.
\end{proof}

\textit{A warning on the notation}. Hereafter in the paper, we shall denote by $C_i$, $i$ an integer index, positive constants which are independent of the parameter $\eps$ and the time $t$.

\begin{lemma}
\label{lem:3}
Let $m_t$ be defined as in \eqref{mt}. For each $\beta\in (0,1/2)$ and $\ell>0$ there exists $\alpha>0$ such that
\begin{equation}
\label{smt}
\lim_{\eps\to 0} \eps^{-\ell} m_t(\eps^\beta) = 0 \quad \forall\, t \in [0,\alpha|\log\eps|]\;. 
\end{equation}
\end{lemma}

\begin{proof}
Given $h>0$, let $x\mapsto W_h(x)$, $x\in \bb R^2$, be a nonnegative smooth function,  depending only on $|x|$, such that
\begin{equation}
\label{W1}
W_h(x) = \begin{cases} 1 & \text{if $|x|\le h$,} \\ 0 & \text{if $|x|\ge 2h$}, \end{cases}
\end{equation}
and, for some $C_1>0$,
\begin{equation}
\label{W2}
|\nabla W_h(x)| < \frac{C_1}{h}\;,
\end{equation}
\begin{equation}
\label{W3}
|\nabla W_h(x)-\nabla W_h(x')| < \frac{C_1}{h^2}\,|x-x'|\;. 
\end{equation}

We define the quantity
\begin{equation}
\label{mass 1}
\mu_t(h) = 1 - \int\! \rmd x \, W_h(x-B_\eps(t))\, \omega_\eps (x,t)\;,
\end{equation}
which is a mollified version of $m_t$, satisfying
\begin{equation}
\label{2mass 3}
\mu_t(h) \le m_t(h) \le \mu_t(h/2)\;.
\end{equation}
In particular, it is enough to prove \eqref{smt} with $\mu_t$ instead of $m_t$. 

To this purpose, we study the time derivative of $\mu_t(h)$. By applying \eqref{WeakeqF} with test function $f(x,t) = W_h(x-B_\eps(t))$, and recalling $u(x,t) = \int\!\rmd y\, K(x-y)\omega_\eps(y,t)$ and \eqref{growth B}, we have, 
\begin{equation}
\label{mass 4}
\begin{split}
\frac{\rmd}{\rmd t} \mu_t(h) & = - \int\! \rmd x\, \nabla W_h(x-B_\eps(t)) \cdot [u(x,t)+ F(x,t) - \dot B_\eps(t)]\,\omega_\eps(x,t) \\ & =  - H_3 - H_4\;, 
\end{split}
\end{equation}
with
\begin{equation*}
\begin{split}
H_3 & = \int\! \rmd x\, \nabla W_h(x-B_\eps(t)) \cdot \int\!\rmd y \, K(x-y)\, \omega_\eps(y,t)\, \omega_\eps(x,t) \\ & = \frac 12 \int\! \rmd x \! \int\! \rmd y\, \omega_\eps(x,t)\,  \omega_\eps(y,t) \, [\nabla W_h(x-B_\eps(t))-\nabla W_h(y-B_\eps(t))] \cdot K(x-y) \;, \\ H_4 & = \int\! \rmd x\, \nabla W_h(x-B_\eps(t)) \cdot \int\!\rmd y \,[F(x,t)-F(y,t)]\, \omega_\eps(y,t)\, \omega_\eps(x,t)\;,
\end{split}
\end{equation*}
where the second expression of $H_3$ is due to the antisymmetry of $K$.

Concerning $H_3$, we introduce the new variables $x'=x-B_\eps(t)$, $y'=y-B_\eps(t)$, and let
\begin{equation*}
F(x',y') = \frac 12 \omega_\eps(x'+B_\eps(t),t)\, \omega_\eps(y'+B_\eps(t),t) \, [\nabla W_h(x')-\nabla W_h(y')] \cdot K(x'-y') \;,
\end{equation*} 
so that $H_3 = \int\!\rmd x' \! \int\!\rmd y'\,F(x',y')$. We observe that $F(x',y')$ is a symmetric function of $x'$ and $y'$ and that, by \eqref{W1}, a necessary condition to be different from zero is if either $|x'|\ge h$ or $|y'|\ge h$. Therefore, 
\begin{equation*}
\begin{split}
H_3  &= \bigg[ \int_{|x'| > h}\!\rmd x' \! \int\!\rmd y' + \int\!\rmd x' \! \int_{|y'| > h}\!\rmd y' -  \int_{|x'| > h}\!\rmd x' \! \int_{|y'| > h}\!\rmd y'\bigg]F(x',y') \\ & = 2 \int_{|x'| > h}\!\rmd x' \! \int\!\rmd y'\,F(x',y')  -  \int_{|x'| > h}\!\rmd x' \! \int_{|y'| > h}\!\rmd y'\,F(x',y') \\ & = H_3' + H_3'' + H_3'''\;.
\end{split}
\end{equation*}
with 
\begin{equation*}
\begin{split}
H_3' & = 2 \int_{|x'| > h}\!\rmd x' \! \int_{|y'| \le h/2}\!\rmd y'\,F(x',y') \;, \quad  H_3'' = 2 \int_{|x'| > h}\!\rmd x' \! \int_{|y'| > h/2}\!\rmd y'\,F(x',y')\;, \\ H_3''' & = -  \int_{|x'| > h}\!\rmd x' \! \int_{|y'| > h}\!\rmd y'\,F(x',y')\;.
\end{split}
\end{equation*}
By the assumptions on $W_h$, we have $\nabla W_h(z) = \eta_h(|z|) z/|z|$ with $\eta_h(|z|) =0$ for $|z| \le h$. In particular, $\nabla W_h(y') = 0$ for $|y'| \le h/2$. Therefore, 
\begin{equation*}
H_3' =  \int_{|x'| > h}\!\rmd x' \, \omega_\eps(x'+B_\eps(t),t) \eta_h(|x'|) \,\frac{x'}{|x'|}\cdot  \int_{|y'| \le h/2}\!\rmd y'\, K(x'-y') \, \omega_\eps(y'+B_\eps(t),t)\;.
\end{equation*}
In view of  \eqref{W2}, $|\eta_h(|z|)| \le C_1/h$, so that 
\begin{equation*}
|H_3'| \le \frac{C_1}{h} m_t(h) \sup_{|x'| > h} \bigg|\frac{x'}{|x'|}\cdot  \int_{|y'| \le h/2}\!\rmd y'\, K(x'-y') \, \omega_\eps(y'+B_\eps(t),t)\bigg| \;.
\end{equation*}
We now observe that the expression inside the modulus in the right-hand side is equal to the term $H_2$ in \eqref{in H_11} (with $h$ in place of $R_t$), which has been bounded in \eqref{H_14} (it is readily seen that the proof works also if the assumption  $|x'|=R_t$ is relaxed to $|x'|\ge R_t$). We conclude that
\begin{equation*}
|H_3'| \le \frac{5C_1I_\eps(t)}{\pi h^4} m_t(h)\;.
\end{equation*}
Now, by \eqref{W3} and then applying the Chebyshev's inequality,
\begin{equation*}
\begin{split}
|H_3''| + |H_3'''| & \le \frac{3C_1}{2\pi h^2} \int_{|x'| \ge h}\!\rmd x' \! \int_{|y'| \ge h/2}\!\rmd y'\,\omega_\eps(y'+B_\eps(t),t)\, \omega_\eps(x'+B_\eps(t),t) \\ & = \frac{3C_1}{2\pi h^2}m_t(h)   \int_{|y'| \ge h/2}\!\rmd y'\,\omega_\eps(y'+B_\eps(t),t)  \le \frac{6C_1 I_\eps(t)}{\pi h^4} m_t(h)\;.
\end{split}
\end{equation*}
In conclusion,
\begin{equation}
\label{h3s}
|H_3| \le  \frac{11C_1 I_\eps(t)}{\pi h^4} m_t(h)\;.
\end{equation}

Concerning $H_4$, we observe that by \eqref{W1} the integrand is different from zero only if $h\le |x-B_\eps(t)|\le 2h$. Therefore, by \eqref{2Lipsc} and \eqref{W2} we have, 
\begin{equation*}
\begin{split}
|H_4| & \le \frac{C_1}{h} 2\|F\|_\infty \int_{|x'|\ge h}\!\rmd x' \omega_\eps(x'+B_\eps(t),t) \int_{|y'|\ge h}\!\rmd y'\,\omega(y'+B_\eps(t),t) \\ & \quad + \frac{C_1}{h} D_t \int_{h \le |x'|\le 2h}\!\rmd x' \omega_\eps(x'+B_\eps(t),t)\int_{|y'| \le h}\!\rmd y'\,|x'- y'| \, \omega(y'+B_\eps(t),t)\;.
\end{split}
\end{equation*}
Since $|x'-y'| \le 3h$ in the domain on integration of the last integral and using the Chebyshev's inequality in the first one we get,
\begin{equation}
\label{h4s}
|H_4| \le \frac{2C_1 \|F\|_\infty I_\eps(t)}{h^3} m_t(h) + 3C_1 D_t m_t(h)\;.
\end{equation}

Given $0<\beta< 1/2$ as in the claim of the lemma, we fix $\beta_* \in (\beta,1/2)$ and choose $\alpha>0$ so small that \eqref{2bound moment} holds with $\delta=2-2D\alpha>4\beta_*$. Then, by \eqref{mass 4}, \eqref{h3s}, \eqref{h4s}, and using that $D_t\le D$, see \eqref{2Lipsc}, we have,
\begin{equation}
\label{2mass 4''}
\frac{\rmd}{\rmd t} \mu_t(h) \le A(h) m_t(h) \quad\forall\, t \in [0,\alpha|\log\eps|]\;,
\end{equation}
where, for some $C_2>0$,
\begin{equation}
\label{mass 4bis}
A(h) = C_2 \bigg(\frac{\eps^\delta}{h^4}+\frac{\eps^\delta}{h^3} + 1\bigg)\;.
\end{equation}
Moreover, there is a constant $A_*>0$ such that $A(h)\le A_*$ for any $h\ge \eps^{\beta_*}$. Therefore, by \eqref{2mass 3} and \eqref{2mass 4''},
\begin{equation}
\label{mass 14'}
\mu_t(h) \le \mu_0(h) + A_* \int_{0}^t \rmd s\, \mu_s(h/2) \quad \forall\,  t \in [0,\alpha|\log\eps|]\quad \forall\, h\ge \eps^{\beta_*}\;,
\end{equation}
which can be iterated $n$ times, provided $2^{-n}h\ge \eps^{\beta_*}$,  so that, for any $t\in [0,\alpha|\log\eps|]$, 
\begin{equation}
\label{mass 15'}
\begin{split}
\mu_t(h) & \le \mu_0(h) + \sum_{j=1}^n \mu_0(2^{-j}h) \frac{(A_*t)^j}{j!} + \frac{A_*^{n+1}}{n!} \int_0^t\!\rmd s\,  (t-s)^n\mu_s(2^{-(n+1)}h) \\ & = \frac{A_*^{n+1}}{n!} \int_0^t\!\rmd s\,  (t-s)^n\mu_s(2^{-(n+1)}h) \le  \frac{(A_*\alpha |\log\eps |)^{n+1}}{(n+1)!}\;,
\end{split}
\end{equation}
where we used that since $\Lambda_\eps(0) \subset \Sigma(z_*|\eps)$ and $\eps <1$ then $\mu_0(2^{-j}h)=0$ for any $j=0,\ldots,n$, and that $\mu_s(2^{-(n+1)}h)\le 1$. By applying \eqref{mass 15'} with $h=\eps^\beta$, $n=\lfloor (\beta_*-\beta)|\log_2\eps| \rfloor$,  and using the Stirling approximation for $(n+1)!$, we obtain that, given $\ell>0$, for $\alpha$ small enough, 
\begin{equation*}
\lim_{\eps\to 0} \eps^{-\ell} \mu_t(\eps^\beta) = 0\quad \forall\, t \in [0,\alpha|\log\eps|]\;,
\end{equation*}
which concludes the proof. 
\end{proof}

\begin{proof}[Proof of Theorem \ref{thm:2}]
By \eqref{2bound moment}, \eqref{hstimv}, \eqref{stimv}, and recalling $D_t\le D$, see \eqref{2Lipsc}, we have, whenever $|x(x_0,t)-B_\eps(t)| = R_t$,
\begin{equation}
\label{stimr}
\frac{\rmd}{\rmd t} |x(x_0,t)- B_\eps(t)| \le 2 D R_t + \frac{20 \eps^\delta}{\pi R_t^3} + \sqrt{\frac{M\eps^{-\nu}\,m_t(R_t/2)}{\pi}} \quad \forall\, t\in [0,\alpha|\log\eps|]\;,
\end{equation}
with $\delta=2-2D\alpha$ and $\alpha>0$ small. This implies that $\Lambda_\eps(t) \subset \Sigma(B_\eps(t)|R(t))$ for any $t\in [0,\alpha|\log\eps|]$, where $R(t)$ is a solution to 
\begin{equation}
\label{stimrbis}
\dot R(t) = 2 D R(t) + \frac{20 \eps^\delta}{\pi R(t)^3} + \sqrt{\frac{M\eps^{-\nu}\,m_t(R(t)/2)}{\pi}}\;, \quad R(0) = \eps\;.
\end{equation}
Indeed, this is true for $t=0$ and, if at some time $t\in (0,\alpha|\log\eps|]$ a fluid particle initially located at  $x_0\in \Lambda_\eps(0)$ reaches the boundary of $\Sigma(B_\eps(t)|R(t))$, then $R(t) = |x(x_0,t)-B_\eps(t)| =  R_t $ necessarily and hence, by \eqref{stimr}, the radial velocity of $x(x_0,t)- B_\eps(t)$ is less than or equal to $\dot R(t)$. 

We now claim that given $\beta'\in (0,1/2)$ there are $\eps'\in (0,1)$ and $\alpha$ such that $R(t)<\eps^{\beta'}$ for any $t\in [0,\alpha|\log\eps|]$ and $\eps\in (0,\eps')$. To prove the claim, we give a proof by contradiction. Let us suppose there is a time $t_1\in (0,\alpha|\log\eps|]$ such that $R(t_1)=\eps^{\beta'}$ and define $t_0 = \inf\{t\in [0,t_1]\colon R(s) > \eps^{\beta_*} \;\forall\, s\in [t,t_1] \}$ with $\beta_*\in (\beta',1/2)$. Then $R(t) \ge \eps^{\beta_*}$ for any $t\in [t_0,t_1]$, which implies $m_t(R(t)/2)\le m_t(\eps^{\beta_*}/2)$ for any $t\in [t_0,t_1]$. We then apply Lemma \ref{lem:3} with $\ell = \nu + 4$ provided $\alpha$ is small enough, so that the last term in the right-hand side of \eqref{stimrbis} is bounded by $\const\eps^2$ on $[t_0,t_1]$. Therefore, we can find a constant $C_3>20/\pi$ such that
\begin{equation*}
\dot R(t) \le 2 D R(t) + C_3 \eps^{\delta-3\beta_*} \quad \forall\, t\in [t_0,t_1]\;.
\end{equation*}
Integrating the above differential inequality we get,
\begin{equation}
\label{stimr2}
\begin{split}
R(t_1) & \le \rme^{2D(t_1-t_0)} \big(R(t_0) + (t_1-t_0)  C_3 \eps^{\delta-3\beta_*}\big)  \\ & \le \eps^{-2D\alpha} \big(\eps^{\beta_*} + C_3 \alpha |\log\eps| \eps^{\delta-3\beta_*}\big)\;.
\end{split}
\end{equation}
As $\delta=2-2D\alpha$, we can choose $\alpha$ so small to have $\min\{\beta_*-2D\alpha;\delta -3\beta_*-2D\alpha\} > \beta'$. Then, there exists $\eps'\in (0,1)$ such that the right-hand side of \eqref{stimr2} is strictly smaller than $\eps^{\beta'}$ for any $\eps\in (0,\eps')$, which contradicts the assumption $R(t_1)=\eps^{\beta'}$.

By the claim just proved and \eqref{2concl2}, we have that $\Lambda_\eps(t)\subset \Sigma(B(t)|r_\eps)$ for any $\eps\in (0,\eps')$, $\alpha$ small enough, and $t\in [0,\alpha|\log\eps|]$, where $r_\eps = \eps^{\beta'} + 2(1+D\alpha|\log\eps|)\eps^{\delta/2}$. Clearly, this concludes the proof of the theorem. Indeed, given $\beta\in (0,1/2)$, by choosing $\beta' \in (\beta, 1/2)$ and $\alpha$ small enough, there exists $\eps_1\in (0,\eps')$ such that $r_\eps<\eps^\beta$ for any $\eps \in (0,\eps_1)$, and hence \eqref{tbeb2} is proved with $\zeta_1=\alpha$.
\end{proof}

\subsection{Proof of Theorem \ref{thm:1}}
\label{sec:2.2}

Theorem \ref{thm:1} follows quite easily from Theorem \ref{thm:2} and we only give a sketch of the proof.  Let $\beta\in (0,1/2)$ be fixed as in the statement of the theorem. We notice that, by continuity, $T_{\eps,\beta}>0$ for any $\eps\in (0,1)$. Recalling \eqref{rmin}, let now $\eps_0'\in (0,1)$ be such that $\mathrm{dist}(\Sigma(z_i(t)|\eps^\beta), \Sigma(z_j(t)|\eps^\beta)) \ge r_\mathrm{min}/2$ for any $t\ge 0$, $i\ne j$, and $\eps\in (0,\eps_0')$. Therefore, for any $\eps\in (0,\eps_0')$ and $t\in [0,T_{\eps,\beta}]$, the blobs evolve with supports $\Lambda_{i,\eps}(t)$ that remain separated by a distance larger than or equal to $r_\mathrm{min}/2$, and hence their mutual interaction remains bounded and Lipschitz. Otherwise stated, during the time interval $[0,T_{\eps,\beta}]$, the $i$-th blob of vorticity $\omega_{i,\eps}(x,t)$ evolves according to a reduced system as in Subsection \ref{sec:2.1}, with $z_*=z_i$ and external field 
\begin{equation}
\label{fk1}
F_{i,\eps}(x,t) = \sum_{\substack{j=1 \\ j\ne i}}^N \int\!\rmd y\, K_1(x,y)\, \omega_{j,\eps}(y,t)\;,
\end{equation}
where $K_1(x,y)$ is any smooth kernel such that $K_1(x,y) = K(x-y)$ if $|x-y|\ge r_\mathrm{min}/2$. Therefore, for some constant $D'>0$,
\begin{equation*}
|F_{i,\eps}(x,t)| \le D' \quad \forall\, (x,t)\in \bb R^2\times [0,T_{\eps,\beta}]\quad\forall\, \eps\in (0,\eps_0')
\end{equation*}
and
\begin{equation*}
|F_{i,\eps}(x,t)-F_{i,\eps}(y,t)| \le  D' |x-y| \quad \forall\, (x,y,t) \in \bb R^2\times\bb R^2\times [0,T_{\eps,\beta}]\quad\forall\, \eps\in (0,\eps_0']\;.
\end{equation*}

It is now easy, with minor adjustments detailed below, to adapt the proof of Theorem \ref{thm:2} and show that there are $\eps_0\in (0,\eps_0']$ and $\zeta_0>0$ such that $\Lambda_{i,\eps}(t) \subset \Sigma(z_i(t)|\eps^\beta)$ for any $\eps\in (0,\eps_0)$, $i=1,\ldots,N$, and $t \le \min\{T_{\eps,\beta};\zeta_0|\log\eps|\}$. By the definition of $T_{\eps,\beta}$, this clearly implies $T_{\eps,\beta}>\zeta_0|\log\eps|$ for any $\eps\in (0,\eps_0)$, thus completing the proof of the theorem. 

The main difference in repeating the analysis of Subsection \ref{sec:2.1} is that here the external fields depend on $\eps$, and are only close to the fields appearing in the right-hand side of the vortex model \eqref{Pointv} in the case $\Gamma=\bb R^2$. This modifies the estimation of the nearness between the centers of vorticity $B_{i,\eps}(t) := a_i^{-1}\int\!\rmd x\, x\, \omega_{i,\eps}(x,t)$ and the corresponding vortices $z_i(t)$, so we discuss only this point. By \eqref{Pointv} and  \eqref{fk1}, for any $i=1,\ldots,N$ and $t \in [0, T_{\eps,\beta}]$,
\begin{equation*}
\begin{split}
\dot B_{i,\eps}(t) - \dot z_i(t) & = a_i^{-1}\int\!\rmd y\, [F_{i,\eps}(y,t) - F_{i,\eps}(B_{i,\eps}(t),t)] \, \omega_{i,\eps}(y,t) \\ & \quad + \sum_{\substack{j=1 \\ j\ne i}}^N \int\!\rmd y\, [K_1(B_{i,\eps}(t),y) - K_1(B_{i,\eps}(t),B_{j,\eps}(t))]\, \omega_{j,\eps}(y,t) \\ & \quad + \sum_{\substack{j=1 \\ j\ne i}}^N a_j \, [K_1(B_{i,\eps}(t),B_{j,\eps}(t)) - K_1(B_{i,\eps}(t),z_j(t))] \\ & \quad + \sum_{\substack{j=1 \\ j\ne i}}^N a_j \, [K_1(B_{i,\eps}(t),z_j(t)) - K_1(z_i(t),z_j(t))]\;,
\end{split}
\end{equation*}
where we used that $K(z_i(t)-z_j(t)) = K_1(z_i(t),z_j(t))$ for $j\ne i$ and $t \in [0, T_{\eps,\beta}]$. Integrating the above identity and arguing as in Lemma \ref{lem:1}, we now obtain, for some constant $C>0$ and any $t \in [0, T_{\eps,\beta}]$,
\begin{equation*}
\Delta(t) \le \Delta(0) \rme^{Ct} + C \rme^{Ct} \int_0^t\!\rmd s\, \sum_{j=1}^N \sqrt{I_{j,\eps}(s)}\;,
\end{equation*}
where $\displaystyle\Delta(t) = \max_{i}|B_{i,\eps}(t)-z_i(t)|$ and $I_{j,\eps}(t) := \int\!\rmd x\, |x-B_{j,\eps}(t)|^2\, \omega_{j,\eps}(x,t)$. This estimate, together with an a priori bound on the moments of inertia $I_{j,\eps}(t)$, gives an estimate like \eqref{2concl2} for $\Delta(t)$. 
\qed

\begin{remark}
\label{rem:2.2}
In the case of a generic domain, for $\eps$ small enough the system remains far from the boundary until the time $T_{\eps,\beta}$, so that also the effect of the boundary can be treated as a regular external (bounded, divergence-free, and Lipschitz) field. A few words on the meaning of the assumption \eqref{rmin} are due in this case. Unlike the case $\Gamma=\bb R^2$, explicit cases where \eqref{rmin} is true are not present in the literature, but it seems very reasonable that this assumption is ``generic'', i.e., it holds almost everywhere. Actually, in presence of boundaries, the global existence of solutions for any choice of initial data and intensities $\{z_i,a_i\}_{i=1}^N$, outside a set of Lebesgue measure zero, has been proved only in the case of a circular domain \cite{MaP84}, but this lack of results does not appear substantial (any regular boundary looks locally like a circle).
\end{remark}

\section{Examples of persistence of vortices on power-law time scales}
\label{sec:3}

In this section we provide examples in which the results of the previous section can be improved, proving that the maximal time for which the blobs of vorticity remain concentrated is not less than an inverse power of the initial size $\eps$ of the blobs.

\subsection{Examples of flow in $\bb R^2$}
\label{sec:3.1}

The simplest example in $\bb R^2$ is given by a blob of vorticity with compact support and alone in the plane. The time goes by and the support could increase. Bounds on the growth are given in \cite{Mar94,Mar96,LoN,Ser,ISG}. In this case, when the vorticity is concentrated around a point, we can obtain a power-law lower bound on the maximal time quoted above. An explicit proof follows by the analysis of the example that we discuss next.

The dynamical system \eqref{Pointv} admits particular choices of  intensities and initial data for which the system evolves in a self-similar configuration, i.e., the polygon with vertices formed by the point vortices rotates and changes its size but remains similar in shape. Denoting by $N$ the number of point vortices, this property has been known for a long time for $N\leq 3$, and more recently for $N=4,5$, and it has been conjectured for larger $N$, see \cite{Are07,NoS,ONe}; some properties have been recently discussed in \cite {IfM}. 
 
For concreteness, we study the case $N=3$, but similar considerations can be made for every $N$. Consider three point vortices of intensities $a_i$ posed in $z_i(t)$. As well known, there exist intensities and initial data for which the bodies approach each other and collide, while for other initial conditions they move away from each other. More precisely, there are intensities and positions for which the three vortices, initially posed on the vertices on a triangle of sides of length $L_{ij}$, in the future remain posed in the vertices of a triangle of side of length $L_{ij}(t)$, where
\begin{equation}
\label{3point}
L_{ij}(t) = L_{ij}(0) \sqrt{1+gt}\;, \quad g>0\;, 
\end{equation}
that is, the triangle grows in the future (and shows a collapse for $t=-g^{-1}$), but remains similar in form. For the time evolution of three point vortices see \cite{Are79} and also \cite{Are07,MaP94}.\footnote{We recall the main conditions for a system of three point vortices to go to infinity \cite{Are79,Are07}. We denote by $(a_1,a_2,a_3)$ the intensities of the vortices and by $L_{ij}$ the distance between vortices $i$ and $j$. There are conditions under which the triangle whose vertexes are given by the positions of the vortices changes size but remains similar in form. These conditions can be easily expressed in terms of the intensities  and the reciprocal distances: $a_1a_2+a_1a_3+a_2a_3=0$ and  $a_1a_2L_{12}^2+a_1a_3L_{13}^2+a_2a_3L_{23}^2=0$. The dynamical system collapses in the past at a critical time and increases its size in the future as the square root of $t$. More precisely, the equations of motion imply
\begin{equation*}
\frac{d}{dt} L^2_{i,j} =\frac{2 A a_k}{\pi}[L^{-2}_{jk}-L^{-2}_{ki}]\;,
\end{equation*}
where $A$ is the area of the triangle determined by the positions of the three vortices with orientation, i.e., reckoned positive if $(i,j,k)$ appear counterclockwise and negative if $(i,j,k)$ appears clockwise. The previous equation implies \eqref{3point}.}

\begin{theorem}
\label{thm:3}
Under the same hypothesis and notation of Theorem \ref{thm:1} with $N=3$, we further assume that the three point vortices evolve according to \eqref{3point}. Then, for each $\beta\in (0,1/2)$ there exist $\eps_0>0$ and $\zeta_0>0$ such that
\begin{equation}
\label{tbeb3}
T_{\eps,\beta} > \eps^{-\zeta_0} \quad\forall\,\eps\in(0,\eps_0)\;. 
\end{equation}
\end{theorem}

The proof is achieved like that of Theorem \ref{thm:1}, i.e., through the analysis of a reduced problem: a single blob of unitary vorticity moving in an external time dependent vector field $F(x,t)$, that simulates the action of the other two blobs of vorticity. Therefore, it is a divergence-free field, with norm and Lipschitz constant decreasing in time, i.e., for suitable constants $b,L>0$,
\begin{equation}
\label{bounded}
|F(x,t)| \le \frac{b}{\sqrt{1+t}}\;,
\end{equation}
\begin{equation}
\label{Lips}
|F(x,t)-F(y,t)| \le D_t |x-y|\;, \quad D_t=\frac{L}{1+t}\;.
\end{equation}

The main point is the decreasing in time of the Lipschitz constant $D_t$, which allows one to improve the content of Theorem \ref{thm:2} in the following way.

\begin{theorem}
\label{thm:4}
Under the same hypothesis and notation of Theorem \ref{thm:2}, we further assume that the external field satisfies \eqref{Lips}. Then, for each $\beta\in (0,1/2)$ there exist $\eps_1>0$ and $\zeta_1>0$ such that
\begin{equation}
\label{tbeb4}
T_{\eps,\beta}^* > \eps^{-\zeta_1} \quad\forall\,\eps\in(0,\eps_1)\;. 
\end{equation}
\end{theorem}

\begin{proof}
We only give a sketch of the proof, by focusing on those parts where the decreasing in time of $D_t$ allows us to improve the result. In what follows, we use the same notation of Section \ref{sec:2}. 

By Lemma \ref{lem:1},
\begin{equation*}
I_\eps(t) \le 4\eps^2 (1+t)^{2L}\;, \quad |B_\eps(t)-B(t)| \le 2\eps[1+L\log(1+t)](1+t)^L\;,
\end{equation*}
which gives, for some constant $C_4>0$, 
\begin{equation}
\label{IBi}
I_\eps(t)  \le \ C_4 \eps^\delta\;, \quad |B(t)-B_\eps(t)| \le C_4 (1+\alpha|\log\eps|) \eps^{\delta/2} \quad \forall\, t \in [0, \eps^{-\alpha}]\;,
\end{equation}
with $\delta = 2-2L\alpha>0$, provided $\alpha$ is small enough. 

We now turn to the content of Lemma \ref{lem:3}, obtaining in this case a better result; more precisely, we assert that for each $\beta\in (0,1/2)$ and $\ell>0$ there exists $\alpha>0$ such that
\begin{equation}
\label{smtt}
\lim_{\eps\to 0} \eps^{-\ell} m_t(\eps^\beta) = 0 \quad \forall\, t \in [0,\eps^{-\alpha}]\;. 
\end{equation}

To this purpose, arguing as in that lemma, we fix $\beta_*\in (\beta,1/2)$ and assume $\alpha>0$ so small that \eqref{IBi} holds with $\delta = 2 - 2\alpha L>0$. We then observe that, by \eqref{h3s} and \eqref{h4s},
\begin{equation}
\label{2masst}
\frac{\rmd}{\rmd t} \mu_t(h) \le A_t(h) m_t(h) \quad\forall\, t \in [0,\eps^{-\alpha}]\;,
\end{equation}
where, for some constant $C_5>0$,
\begin{equation*}
A_t(h) = C_5\bigg(\frac{\eps^\delta}{h^4} +\frac{\eps^\delta}{h^3} + \frac{1}{1+t}\bigg) \le C_5\bigg(2\eps^{\delta-4\beta_*} + \frac{1}{1+t}\bigg) \quad \forall\, h\ge \eps^{\beta_*}\;.
\end{equation*}
We choose $\alpha>0$ and $\eps_*\in (0,1)$ small enough to have 
\begin{equation*}
A_t(h) \le \frac{2C_5}{1+t} \quad \forall\, t \in [0,\eps^{-\alpha}] \quad \forall\, h\ge \eps^{\beta_*} \quad \forall\, \eps \in (0,\eps_*)\;,
\end{equation*}
and hence,
\begin{equation*}
\mu_t(h) \le \mu_0(h) + 2C_5 \int_{0}^t \rmd s\, \frac1{1+s}\, \mu_s(h/2) \quad \forall\,  t \in [0,\eps^{-\alpha}]\quad \forall\, h\ge \eps^{\beta_*}\;,
\end{equation*}
which can be iterated $n$ times, provided $2^{-n}h\ge \eps^{\beta_*}$, as done in \eqref{mass 15'}, getting now that, for any $t\in [0,\eps^{-\alpha}]$, 
\begin{equation*}
\mu_t(h) \le (2C_5)^n \int_0^t\!\frac{\rmd s_1}{1+s_1}\cdots \int_0^{s_{n}}\! \frac{\rmd s_{n+1}}{1+s_{n+1}}\mu_{s_{n+1}}(h/2^{n+1}) \le \frac{[2C_5\log(1+t)]^{n+1}}{(n+1)!}\;.
\end{equation*}
Choosing $h=\eps^\beta$, $n=\lfloor (\beta_*-\beta)|\log_2\eps| \rfloor$ and using the Stirling approximation for $(n+1)!$, we obtain that for any $\ell>0$ there is $\alpha$ small enough such that $\lim_{\eps\to 0} \eps^{-\ell} \mu_t(\eps^\beta) = 0$ for any $t \in  [0, \eps^{-\alpha}]$, which implies \eqref{smtt} in view of \eqref{2mass 3}.

By \eqref{hstimv} and \eqref{stimv}, using $D_t = L(1+t)^{-1}$ and the bound \eqref{IBi} on $I_\eps(t)$,  we have now that, whenever $|x(x_0,t)-B_\eps(t)| = R_t$,
\begin{equation*}
\frac{\rmd}{\rmd t} |x(x_0,t)- B_\eps(t)| \le \frac{2L}{1+t} R_t + \frac{5C_4\eps^\delta}{\pi R_t^3} + \sqrt{\frac{M\eps^{-\nu}\,m_t(R_t/2)}{\pi}} \quad \forall\, t\in [0,\eps^{-\alpha}]\;,
\end{equation*}
with $\delta=2-2L\alpha$ and $\alpha>0$ small. Reasoning as in Theorem \ref{thm:2}, this implies that $\Lambda_\eps(t) \subset \Sigma(B_\eps(t)|R(t))$ for any $t\in [0,\eps^{-\alpha}]$, where $R(t)$ is a solution to 
\begin{equation}
\label{stirmtris}
\dot R(t) = \frac{2L}{1+t}  R(t) + \frac{5C_4\eps^\delta}{\pi R(t)^3} + \sqrt{\frac{M\eps^{-\nu}\,m_t(R(t)/2)}{\pi}}\;, \quad R(0) = \eps\;.
\end{equation}

We now show that given $\beta'\in (0,1/2)$ there are $\eps'\in (0,1)$ and $\alpha$ such that $R(t)<\eps^{\beta'}$ for any $t\in [0,\eps^{-\alpha}]$ and $\eps\in (0,\eps')$. We give a proof by contradiction. Let us suppose there is a time $t_1\in (0,\eps^{-\alpha}]$ such that $R(t_1)=\eps^{\beta'}$ and define $t_0 = \inf\{t\in [0,t_1]\colon R(s) > \eps^{\beta_*} \;\forall\, s\in [t,t_1] \}$  with $\beta_*\in (\beta',1/2)$. Then $R(t) \ge \eps^{\beta_*}$ for any $t\in [t_0,t_1]$, which implies $m_t(R(t)/2)\le m_t(\eps^{\beta_*}/2)$ for any $t\in [t_0,t_1]$. We then apply \eqref{smtt} with $\ell = \nu + 4$ provided $\alpha$ is small enough, so that the last term in the right-hand side of \eqref{stirmtris} is bounded by $\const\eps^2$ on $[t_0,t_1]$. Therefore, we can find a constant $C_6>5C_4/\pi$ such that
\begin{equation*}
\dot R(t) \le \frac{2L}{1+t}  R(t) + C_6 \eps^{\delta-3\beta_*} \quad \forall\, t\in [t_0,t_1]\;.
\end{equation*}
Integrating the above differential inequality we get,
\begin{equation}
\label{stimr3}
\begin{split}
R(t_1) & \le \bigg(\frac{1+t_1}{1+t_0} \bigg)^{2L} \big(R(t_0) + (t_1-t_0)  C_6 \eps^{\delta-3\beta_*}\big)  \\ & \le (1+\eps^{-\alpha})^{2L} \big(\eps^{\beta_*} + C_6 \eps^{\delta-3\beta_*-\alpha}\big)\;.
\end{split}
\end{equation}
As $\delta=2-2L\alpha$, we can choose $\alpha$ so small to have $\min\{\beta_*-2L\alpha;\delta -3\beta_*-(2L+1) \alpha\} > \beta'$. Then, there exists $\eps'\in (0,1)$ such that the right-hand side of \eqref{stimr3} is strictly smaller than $\eps^{\beta'}$ for any $\eps\in (0,\eps')$, which contradicts the assumption $R(t_1)=\eps^{\beta'}$.

We can now conclude as in the proof of Theorem \ref{thm:2}. Given $\beta\in (0,1/2)$, by choosing $\beta' \in (\beta, 1/2)$ and $\alpha$ small enough, from the above estimate on $R(t)$ and the second estimate in \eqref{IBi} it follows that there is $\eps_1\in (0,\eps')$ such that $\Lambda_\eps(t)\subset \Sigma(B(t)|\eps^\beta)$ for any $t\in [0,\eps^{-\alpha}]$ and $\eps\in (0,\eps_1)$, thus proving \eqref{tbeb4} with $\zeta_1=\alpha$.
\end{proof}

We omit the proof of how to deduce Theorem \ref{thm:3} from Theorem \ref{thm:4}: it can be easily obtained by adapting to the present context the proof of Theorem \ref{thm:1} discussed in Subsection \ref{sec:2.2}. 

\subsection{An example of flow in a bounded region}
\label{sec:3.2}

Consider a single blob of vorticity in a bounded domain $\Gamma$, initially concentrated around a point $z_0\in \Gamma$. In this case, denoting by $m$ the total mass of the blob, the corresponding point vortex system \eqref{Pointv} reads, 
\begin{equation}
\label{point Gamma}
\dot{z}(t) = \frac m2 \nabla^\perp \gamma(z(t)) \;, \quad z(0)=z_0\;,
\end{equation}
where $\gamma (x) = \gamma_\Gamma (x,x)$, with $\gamma_\Gamma (x,y) =  G_\Gamma (x,y) + \frac{1}{2\pi} \log |x-y|$ and $ G_\Gamma$ the fundamental solution of the Laplace operator in $\Gamma$ vanishing on the boundary.

In what follows, we assume that $\Gamma=\Sigma(0,1)$, the disk of unitary radius centered in the origin. In this case, $G_\Gamma$ is the sum of the Green function in the whole plane plus a contribution given by a negative mirror charge posed in $\bar y = y/|y|^2$, and hence $\gamma_\Gamma(x,y) = \frac{1}{2\pi} \log|x-\bar y|$. 

If initially the blob of vorticity is contained in $\Sigma(0,\eps)$, the time evolution is similar to that of a blob which moves in the whole plane and it is subjected to an external field  which is vanishing as $\eps \to 0$. The origin is an equilibrium for the vortex dynamics \eqref{point Gamma}, and the states with radial symmetric vorticity distribution are stationary for the Euler dynamics. Instead, a blob of vorticity without such symmetry is not stationary in general, and small filaments of vorticity can move away. Nonetheless, by adapting the techniques of previous sections, we can prove that the blob remains concentrated up to times of the order of an inverse power of $\eps$.

\begin{theorem}
\label{thm:5}
Let $\Gamma=\Sigma(0,1)$ and assume that the initial datum $\omega_\eps(x,0)$ of the Euler equations is given by a single blob of vorticity with support $\Lambda_\eps(0) \subset \Sigma(0|\eps)$. Suppose also that there are $M,\nu>0$ such that
\begin{equation*}
|\omega_{i,\eps}(x,0)| \le M \eps^{-\nu}\;.
\end{equation*}
Then, for each $\beta\in (0,1/2)$ there exist $\eps_0>0$ and $\zeta_0>0$ such that
\begin{equation}
\label{tbeb5}
T_{\eps,\beta} > \eps^{-\zeta_0} \quad\forall\,\eps\in(0,\eps_0)\;,
\end{equation}
with $T_{\eps,\beta}$ as defined in \eqref{tbe}, which in this case reduces to
\begin{equation*}
T_{\eps,\beta} =\sup\{t>0 \colon |x(x_0,s)| < \eps^\beta\;\; \forall\, s\in [0,t]\;\forall\, x_0\in \Lambda_\eps(0)\}\;.
\end{equation*}
\end{theorem}

\begin{proof}
Without lack of generality we assume that the total mass of vorticity equals one.

Each fluid particle moves in the velocity field produced by the vorticity via  $K_\Gamma (x,y)$,
\begin{equation}
\label{v.field}
K_\Gamma (x,y) = -\nabla_x^\perp \frac{1}{2\pi} \log |x-y| + \nabla_x^\perp \frac{1}{2\pi} \log |x-\bar{y}|\;,
\end{equation}
with $\bar y = y/|y|^2$ as previously defined. Therefore, as long as $\Lambda_\eps(t)$ does not intersect the boundary of the disk $\Gamma$ (in particular, up to time $T_{\eps,\beta}$), the vorticity $\omega_\eps(x,t)$ evolves as in the reduced problem in $\bb R^2$ defined by \eqref{Cons4}, \eqref{WeakeqF}, with the role of external field played by 
\begin{equation}
\label{F}
F(x,t) = \int_\Gamma\! \rmd y\, F_1(x,y)\, \omega_\eps(y,t)\;, 
\end{equation}
where
\begin{equation}
\label{F1}
F_1(x,y) = \nabla_x^\perp \frac{1}{2\pi} \log |x-\bar{y}|\;.
\end{equation}
Since $\bar y = y/|y|^2$, a straightforward computation shows that there is $\kappa>0$ such that, for any $\delta \in (0,1/2)$,
\begin{equation}
\label{4prop.3}
|F_1(x,y)-F_1(z,y)| \le \kappa \delta^2 |x-z| \quad \forall\, x,y,z\in \Sigma(0|\delta) \;,
\end{equation}
which is the key property of $F_1(x,y)$ that will be used in the subsequent analysis.

We study the dynamics in the time interval $[0,T_{\eps,\beta}\wedge\eps^{-\alpha}]$, $\alpha>0$. Our task is to prove that there is $\alpha$ such that $\Lambda_\eps(t) \subset \Sigma(0|\eps^\beta)$ for any $t\in [0,T_{\eps,\beta}\wedge\eps^{-\alpha}]$ and $\eps$ small enough, which implies, by continuity, $T_{\eps,\beta}> \eps^{-\alpha}$.

Hereafter, we assume $\eps$ small enough to have $\eps^\beta\le 1/2$, so that, from \eqref{F} and \eqref{4prop.3},
\begin{equation*}
|F(x,t)-F(z,t)| \le \kappa \eps^{2\beta} |x-z| \quad \forall\, x,y \in \Lambda_\eps(t)\quad\forall\, t\in  [0,T_{\eps,\beta}\wedge\eps^{-\alpha}]\;.
\end{equation*}
This implies that we can apply Lemma \ref{lem:1} with $D_t= \kappa \eps^{2\beta}$ for any $t\in  [0,T_{\eps,\beta}\wedge\eps^{-\alpha}]$. Indeed, both the proofs of the estimates for $I_\eps(t)$ and $|B_\eps(t)-B(t)|= |B_\eps(t)|$ require the control of $|F(x,t)-F(y,t)|$ only for $x,y\in \Lambda_\eps(t)$. Therefore,
\begin{equation*}
I_\eps(t) \le 4\eps^2 \rme^{2\kappa\eps^{2\beta} t}\;, \quad |B_\eps(t)| \le 2\eps(1+\kappa\eps^{2\beta} t) \rme^{\kappa\eps^{2\beta} t} \quad \forall\, t\in  [0,T_{\eps,\beta}\wedge\eps^{-\alpha}]\;,
\end{equation*}
from which we get, assuming $\alpha < 2\beta$, 
\begin{equation}
\label{IBi2}
I_\eps(t)  \le \ 4 \rme^{2\kappa} \eps^2\;, \quad |B_\eps(t)| \le  2(1+\kappa) \rme^{2\kappa} \eps \quad \forall\, t\in  [0,T_{\eps,\beta}\wedge\eps^{-\alpha}]\;.
\end{equation}

By the first estimate in \eqref{IBi2}, the analysis of Lemma \ref{lem:3} can be repeated in this case, obtaining the following result: for each $\beta\in (0,1/2)$ there exists $\alpha>0$ such that
\begin{equation}
\label{smtt2}
\lim_{\eps\to 0} \sup_{\ell>0} \eps^{-\ell} m_t(\eps^\beta) = 0 \quad \forall\, t \in [0,T_{\eps,\beta}\wedge\eps^{-\alpha}]\;. 
\end{equation}

Indeed, arguing as in the proof of that lemma, in view of \eqref{h3s} and \eqref{h4s} we have,
\begin{equation*}
\frac{\rmd}{\rmd t} \mu_t(h) \le A(h) m_t(h) \quad\forall\, t \in [0,T_{\eps,\beta}\wedge\eps^{-\alpha}]\;,
\end{equation*}
where, in this case, for some constant $C_7>0$,
\begin{equation*}
A(h) = C_7 \bigg(\frac{\eps^2}{h^4} +\frac{\eps^2}{h^3} +\eps^{2\beta} \bigg) \le 3 C_7 \eps^\eta \quad \forall\, h\ge \eps^{\beta_*}\;,
\end{equation*}
provided $\eta\le \min\{2\beta, 2-4\beta_*\}$. Hence,
\begin{equation*}
\mu_t(h) \le \mu_0(h) + 3C_7\eps^\eta \int_{0}^t \rmd s\, \mu_s(h/2) \quad \forall\,  t \in [0,T_{\eps,\beta}\wedge\eps^{-\alpha}]\quad \forall\, h\ge \eps^{\beta_*}\;,
\end{equation*}
which can be iterated $n$ times, provided $2^{-n}h\ge \eps^{\beta_*}$, as done in \eqref{mass 15'}, getting now, for any $t\in [0,T_{\eps,\beta}\wedge\eps^{-\alpha}]$, 
\begin{equation*}
\mu_t(h) \le \frac{(3C_7\eps^\eta t)^{n+1}}{(n+1)!} \le \frac{(3C_7\eps^{\eta-\alpha})^{n+1}}{(n+1)!} \;.
\end{equation*}
Choosing $h=\eps^\beta$, $n=\lfloor (\beta_*-\beta)|\log_2\eps| \rfloor$ and using the Stirling approximation for $(n+1)!$, we obtain that there is $\alpha$ small enough such that $\lim_{\eps\to 0} \eps^{-\ell} \mu_t(\eps^\beta) = 0$ for any $\ell>0$ and $t \in [0,T_{\eps,\beta}\wedge\eps^{-\alpha}]$, which implies \eqref{smtt2} in view of \eqref{2mass 3}.

Now, by \eqref{hstimv} and \eqref{stimv}, using $D_t = \kappa\eps^{2\beta}$ and the bound \eqref{IBi2} on $I_\eps(t)$,  we have that, whenever $|x(x_0,t)-B_\eps(t)| = R_t$ and $t\in [0,T_{\eps,\beta}\wedge\eps^{-\alpha}]$,
\begin{equation*}
\frac{\rmd}{\rmd t} |x(x_0,t)- B_\eps(t)| \le 2\kappa\eps^{2\beta} R_t + \frac{20 \rme^{2\kappa} \eps^2}{\pi R_t^3} + \sqrt{\frac{M\eps^{-\nu}\,m_t(R_t/2)}{\pi}}\;.
\end{equation*}
Arguing as in the proof of Theorem \ref{thm:2}, this implies that $\Lambda_\eps(t) \subset \Sigma(B_\eps(t)|R(t))$ for any $t\in [0,T_{\eps,\beta}\wedge\eps^{-\alpha}]$, where $R(t)$ is a solution to 
\begin{equation}
\label{stirmquadris}
\dot R(t) = 2\kappa\eps^{2\beta} R(t) + \frac{20 \rme^{2\kappa} \eps^2}{\pi R(t)^3} + \sqrt{\frac{M\eps^{-\nu}\,m_t(R(t)/2)}{\pi}}\;, \quad R(0) = \eps\;.
\end{equation}

We now show that given $\beta'\in (0,1/2)$ there are $\eps'\in (0,1)$ and $\alpha$ such that $R(t)<\eps^{\beta'}$ for any $t\in [0,T_{\eps,\beta}\wedge\eps^{-\alpha}]$ and $\eps\in (0,\eps')$. By contradiction, let us suppose there is a time $t_1\in (0,\eps^{-\alpha}]$ such that $R(t_1)=\eps^{\beta'}$ and define $t_0 = \inf\{t\in [0,t_1]\colon R(s) > \eps^{\beta_*} \;\forall\, s\in [t,t_1] \}$  with $\beta_*\in (\beta',1/2)$. 
Then $R(t) \ge \eps^{\beta_*}$ for any $t\in [t_0,t_1]$, which implies $m_t(R(t)/2)\le m_t(\eps^{\beta_*}/2)$ for any $t\in [t_0,t_1]$. We then apply \eqref{smtt2} with $\ell = \nu + 4$ provided $\alpha$ is small enough, so that the last term in the right-hand side of \eqref{stirmquadris} is bounded by $\const\eps^2$ on $[t_0,t_1]$. Therefore, we can find a constant $C_8>0$ such that
\begin{equation*}
\dot R(t) \le 2\kappa\eps^{2\beta}  R(t) + C_8 \eps^{2-3\beta_*} \quad \forall\, t\in [t_0,t_1]\;,
\end{equation*}
which implies,
\begin{equation}
\label{stimr4}
\begin{split}
R(t_1) & \le \rme^{2\kappa\eps^{2\beta}(t_1-t_0)} \big(R(t_0)+(t_1-t_0)C_8 \eps^{2-3\beta_*}\big) \\ & \le \rme^{2\kappa\eps^{2\beta-\alpha}} \big(\eps^{\beta_*} + C_7 \eps^{2-3\beta_*-\alpha}\big)\;.
\end{split}
\end{equation}
As $2-3\beta_*>1/2 >\beta'$, there exist $\alpha>0$ and $\eps'\in (0,1)$ such that the right-hand side of \eqref{stimr4} is strictly smaller than $\eps^{\beta'}$ for any $\eps\in (0,\eps')$, which contradicts the assumption $R(t_1)=\eps^{\beta'}$.

We can now conclude the proof of the theorem. Given $\beta\in (0,1/2)$, by choosing $\beta' \in (\beta, 1/2)$ and $\alpha$ small enough, from the above estimate on $R(t)$ and the second estimate in \eqref{IBi2} it follows that there is $\eps_0\in (0,\eps')$ such that $\Lambda_\eps(t)\subset \Sigma(0|\eps^\beta)$ for any $t\in [0,T_{\eps,\beta}\wedge\eps^{-\alpha}]$ and $\eps\in (0,\eps_0)$, thus proving \eqref{tbeb5} with $\zeta_0=\alpha$.
\end{proof}

\section{Analysis of a toy model}
\label{sec:4}

As discussed before, it is too difficult to improve rigorously the estimate \eqref{tbeb1} in general. The difficult task is to control the self-energy of a blob. Perhaps, some better results could be obtained by well-preparing the initial state: blobs of vorticity with a radial symmetry. In this case, the self energy is exactly zero and this trick is used in \cite{Bat} to justify the point vortex model. Of course, the time evolution destroys this symmetry and the justification is weak. However, we can hope that these bad effects, as the initial concentration is large, become important only after a long time, but there is not rigorous proof of this. 

To have some hint in the study of the problem, we introduce a very schematic toy model: a point $x$ in the plane moving under the action of two velocity fields $F(x,t)$ and $g(x,B(t))$, where $B(t)$ is a solution of the equation,
\begin{equation}
\label{eF}
\dot B(t) = F(B(t),t)\;.
\end{equation}

The field $F$ simulates the action of other vortices and it is assumed smooth and divergence-free, the field $g$ simulates the action of the vorticity belonging to the blob itself,
\begin{equation}
\label{g}
g(x,B(t))= - \nabla^\perp \log |x- B(t)| = -\frac{(x-B(t))^\perp}{(x-B(t))^2}\;,
\end{equation}
where, to simplify the notation, we assumed the intensity of the blob equals to $2\pi$. We refer the reader to Remark \ref{rem:4.2} below for more details on the justification of the choice \eqref{g}. 

\begin{theorem}
\label{thm:6}
Let $F(x,t) \in \bb R^2$, $(x,t) \in \bb R^2\times\bb R$, be a (time-dependent) smooth vector field such that $\mathrm{div}\,F(x,t) = 0$. We also assume that $F$, $\partial_t F$, $D_xF$, $\partial_t D_xF$, $D_x^2F$, $\partial_t D_x^2F$, and $D_x^3F$ are uniformly bounded. Given $z_*\in \bb R^2$, we denote by $t\mapsto B(t)$, $t\in\bb R$, the solution to \eqref{eF} with initial condition $B(0) = z_*$.

Let $x(t)$, $t\in I$, be the maximal solution to the Cauchy problem,
\be
\label{p:1}
\dot x(t) = F(x(t),t) - \nabla^\perp \log |x(t)- B(t)|\;, \quad x(0)=x_0\;,
\ee
with $\eps:=|x_0 - z_*|>0$. For each $\beta\in (0,1)$ there exist $\eps_0\in (0,1]$ and $c_0>0$ such that, for any $\eps\in (0,\eps_0)$ we have,
\begin{itemize}
\item[(i)] $[0,\eps^{-\beta}]\subset I$;
\item[(ii)] $|x_0 - z_*| -  c_0 \eps^{3-\beta} \le |x(t)-B(t)| \le |x_0 - z_*| + c_0 \eps^{3-\beta}$ for any $t\in [0,\eps^{-\beta}]$.
\end{itemize}
\end{theorem}

\begin{proof}
A straightforward computation shows that the problem \eqref{p:1} written in term of the unknown $\xi(t):=\eps^{-1}(x(t) - B(t))$ reads,
\be
\label{p:2}
\dot \xi(t) = A(t) \xi(t) + \frac\eps 2 \, \xi(t) \cdot H(t) \xi(t) + \eps^2 Q_\eps(\xi(t),t) - \frac{\xi(t)^\perp}{\eps^2 |\xi(t)|^2} \;, \quad \xi(0)=\xi_0\;,
\ee
where 
\begin{equation*}
A(t) := D_xF(B(t),t)\;, \quad H(t) := D_x^2 F(B(t),t) = \begin{pmatrix} H^{(1)}(t) \\ H^{(2)} (t) \end{pmatrix}\;, 
\end{equation*}
with $H^{(i)}(t) = D_x^2 F_i(B(t),t)$, $i=1,2$, and the short notation $\xi(t) \cdot H(t) \xi(t)$ means
\begin{equation*}
\xi(t) \cdot H(t) \xi(t) = \begin{pmatrix} \xi(t) \cdot H^{(1)}(t) \xi(t) \\ \xi(t)\cdot H^{(2)}(t) \xi(t) \end{pmatrix}\;.
\end{equation*}
Finally, the remainder
\begin{equation*}
Q_\eps(\xi,t) := \eps^{-2}\bigg[F(B(t)+\eps \xi,t) - F(B(t),t) - \eps A(t) \xi - \frac\eps 2 \, \xi \cdot H(t) \xi \bigg]
\end{equation*}
is a smooth function of $\xi$, $t$, and $\eps$, such that $Q_\eps(\xi,t)=\mc O (|\xi|^3)$
(uniformly with respect to $(t,\eps)\in\bb R\times [0,1]$).

Since $\mathrm{div}\,F = 0$ the matrix $A(t)$ is traceless, i.e., it has the form $\begin{pmatrix}a(t) & b(t) \\ c(t) & -a(t) \end{pmatrix}$. Moreover, by differentiating the identity $\mathrm{div}\,F = 0$ with respect to the spatial variables, we have that 
\begin{equation*}
\frac{\partial^2 F_1}{\partial x_1\partial x_2} = - \frac{\partial^2 F_2}{\partial x_2^2}\;, \quad \frac{\partial^2 F_2}{\partial x_1\partial x_2} = - \frac{\partial^2 F_1}{\partial x_1^2}\;,
\end{equation*}
so that the matrices $H^{(i)}(t)$, $i=1,2$ take the form,
\begin{equation*}
H^{(1)}(t) = \begin{pmatrix} h(t) & - p(t) \\ - p(t) & q(t) \end{pmatrix}\;, \quad 
H^{(2)}(t) = \begin{pmatrix} r(t) & - h(t) \\ - h(t) & p(t) \end{pmatrix}\;.
\end{equation*}

The dynamics defined by \eqref{p:2} is characterized by the fact that the polar angle of $\xi$ evolves in time much more quickly than the radius $|\xi|$. By the averaging principle, we then expect the existence of a slowly varying quantity which remains close to the radius $|\xi|$. This quantity can be easily identified by standard perturbation theory, but, for the sake of brevity, we prefer to give its explicit expression without justification, and then proving, by direct inspection, that it is slowly varying in time when computed along the solutions to \eqref{p:2}. With this premise, we introduce the following quantity
\be
\label{p:3}
\varrho(\xi,t) := |\xi| \bigg[1 - \frac{\eps^2}2 \xi^\perp \cdot A(t) \xi + \frac{\eps^3}{2} \bigg(\frac{r(t)}3 \xi_1^3 -  h(t) \xi_1^2\xi_2 + p(t) \xi_1\xi_2^2 - \frac{q(t)}3 \xi_2^3 \bigg)\bigg]
\ee
and compute its time derivative along the solution to \eqref{p:2}. We have (omitting the explicit dependence on $t$),
\begin{equation*}
\begin{split}
\dot \varrho & = \frac{\xi\cdot \dot \xi}{|\xi|} \bigg[1 - \frac{\eps^2}2 \xi^\perp \cdot A\xi + \frac{\eps^3}{2} \bigg(\frac{r}3 \xi_1^3 -  h \xi_1^2\xi_2 + p \xi_1\xi_2^2 - \frac{q}3 \xi_2^3  \bigg)\bigg] \\ & \quad - \frac{\eps^2}2 |\xi|\, \xi^\perp \cdot\dot A \xi - \frac{\eps^2}2 |\xi| \bigg(\xi^\perp \cdot A \dot \xi + \dot \xi^\perp \cdot A\xi \bigg) \\ & \quad + \frac{\eps^3}{2}|\xi| \bigg(\frac{\dot r}3 \xi_1^3 -  \dot h \xi_1^2\xi_2 + \dot p \xi_1\xi_2^2 - \frac{\dot q}3 \xi_2^3  \bigg) \\ &\quad + \frac{\eps^3}{2}|\xi| \big(r \xi_1^2\dot\xi_1 - h\xi_1^2\dot\xi_2 - 2 h \xi_1\xi_2\dot \xi_1 + p\dot\xi_1\xi_2^2 + 2p \xi_1\xi_2\dot\xi_2 -q\xi_2^2\dot\xi_2 \big)\;.
\end{split}
\end{equation*}
Using the right-hand side of \eqref{p:2} to express the time derivatives of $\xi(t)$, after some straightforward computations we get,
\begin{equation*}
\begin{split}
 \\ \dot \varrho & = \frac{\xi \cdot A\xi}{|\xi|} + \eps\frac{\xi\cdot (\xi\cdot H\xi)}{2|\xi|} + \frac{\xi^\perp \cdot A \xi^\perp}{2|\xi|} + \frac{(\xi^\perp)^\perp \cdot A \xi}{2|\xi|} \\ & \quad - \frac{\eps}{2|\xi|} \big(r \xi_1^2\xi_2 +h\xi_1^3 - 2 h \xi_1\xi_2^2 +p\xi_2^3 - 2p \xi_1^2\xi_2 + q\xi_1\xi_2^2 \big) + \eps^2 P_\eps(\xi)\;,
\end{split}
\end{equation*}
where $P_\eps(\xi)$ is a smooth function of $\xi$, $t$, and $\eps$, such that $P_\eps(\xi,t)=\mc O (|\xi|^3)$.

Using that $A$ is traceless, a calculation shows that  $(A\xi)^\perp = - A^T \xi^\perp$, with $A^T$ the transpose of $A$. Therefore, noticing also that $(\xi^\perp)^\perp = - \xi$, 
\begin{equation*}
\begin{split}
 \xi^\perp \cdot A \xi^\perp = A^T \xi^\perp \cdot \xi^\perp = - (A\xi)^\perp \cdot \xi^\perp  = - \xi \cdot A\xi\;,\qquad (\xi^\perp)^\perp \cdot A \xi = - \xi \cdot A \xi\;.
\end{split}
\end{equation*}
Moreover,
\begin{equation*}
\begin{split}
\xi\cdot (\xi\cdot H\xi) & = \xi_1(h\xi_1^2+q\xi_2^2-2p\xi_1\xi_2) + \xi_2(r\xi_1^2+p\xi_2^2-2h\xi_1\xi_2) \\ & = r \xi_1^2\xi_2 +h\xi_1^3 - 2 h \xi_1\xi_2^2 +p\xi_2^3 - 2p \xi_1^2\xi_2 + q\xi_1\xi_2^2\;.
\end{split}
\end{equation*}
Inserting the above identities in the expression for $\dot\varrho$ we obtain,
\be
\label{p:4}
\dot \varrho = \eps^2 P_\eps(\xi) \;.
\ee

Now we set
\begin{equation*}
\tau := \sup\{t\in I\cap \bb R_+\colon|\xi(s)|<2\;\; \forall\, s\in [0,t]\}\;, 
\end{equation*}
noticing that $\tau>0$ since $|\xi_0|=1$ and $\xi(t)$ is continuous. By \eqref{p:3} and \eqref{p:4} there are $\eps_1\in (0,1]$ and $C>0$ such that, for any $\eps\in (0,\eps_1]$ and $t\in [0,\tau)$,
\be
\label{p:5}
|\varrho(t) - |\xi(t)|| \le C \eps^2 |\xi(t)|\;,\quad \frac 12 |\xi(t)| \le \varrho(t) \le \frac 32 |\xi(t)|\;, \quad |\dot \varrho(t)| \le C\eps^2 |\xi(t)|
\ee
(the second inequality comes from the first one for $2C\eps_1^2<1$).  From the above estimates we obtain the following integral inequality, valid for any $\eps\in (0,\eps_1]$,
\begin{equation*}
|\varrho(t)-\varrho(0)| \le 2C\eps^2 \tau \varrho(0) + 2C\eps^2 \int_0^t\!\rmd s\,|\varrho(s)-\varrho(0)| \quad \forall\,t\in [0,\tau)\;,
\end{equation*}
which implies, 
\be
\label{p:6}
\max_{s\in [0,t]} |\varrho(s)-\varrho(0)| \le 2C \eps^2\tau \varrho(0) + 2C\eps^2\tau\max_{s\in [0,t]} |\varrho(s)-\varrho(0)| \quad \forall\,t\in [0,\tau)\;.
\ee
Given $\beta\in (0,1)$ we choose $\eps_0\in (0,\eps_1]$ such that $K_0:=12C\eps_0^{2-\beta} < 1$. Then, setting $\tau_0 = \tau\wedge\eps^{-\beta}$, from \eqref{p:6} we deduce that, for any $\eps\in (0,\eps_0)$,
\begin{equation*}
|\varrho(t)-\varrho(0)| \le 4C \eps^{2-\beta} \varrho(0) \quad \forall\,t\in [0,\tau_0)\;.
\end{equation*}
Moreover, recalling $|\xi_0|=1$ and using the first two inequalities in \eqref{p:5}, for any $\eps\in (0,\eps_0)$,
\begin{equation*}
\begin{split}
||\xi(t)|-1| & \le ||\xi(t)|-\varrho(t)| + |\varrho(t)-\varrho(0)|+|\varrho(0)-|\xi(0)|| \\ & \le 2C\eps^2 \varrho(t) +  4C \eps^{2-\beta} \varrho(0) + 2C\eps^2\varrho(0) \\ & \le 6C(\eps^2 + \eps^{2-\beta}) \le K_0<1 \quad \forall\,t\in [0,\tau_0)\;.
\end{split}
\end{equation*}
This implies $\tau_0\in I$ and, by continuity, $|\xi(\tau_0)| < 2$. From the definition of $\tau$ we conclude that $\tau_0 = \eps^{-\beta}$ for any $\eps\in (0,\eps_0)$. As $|x(t)-B(t)| - |x_0 - z_*| = \eps (|\xi(t)|-1)$, this proves both claims (i) and (ii) of the theorem.
\end{proof}

\begin{remark}
\label{rem:4.1}
It is worthwhile to remark that the condition $\mathrm{div}\,F(x,t)=0$ plays a crucial role for the validity of the above result. Indeed, if the matrix $ A (t) $ is not traceless, the averaging effect does not cancel, in general, the lowest order of the expansion of $\dot \varrho$, due to the possible presence of secular terms in the averaged system. Moreover, the not trivial fact is that such cancellation turns out to be valid up to the second order and not only to the first one as expected in general. As we shall see in the next remarks, this fact allow one to enhance the justification of the point vortex model via radially symmetry given in textbooks such as \cite{Bat}.   
\end{remark}

\begin{remark}
\label{rem:4.1b}
We observe that from the proof of the theorem, more precisely in view of  the estimates \eqref{p:5} and \eqref{p:6}, it is easy to deduce also that if $c>0$ is chosen sufficiently small then $|x(t)-B(t)| < 2 \eps$ for any  $t\in [0,c\,\eps^{-2}]$ and any $\eps$ small enough. On the other hand, only the better estimates obtained in the shorter time interval $[0,\eps^{-\beta}]$, $\beta\in (0,1)$, as stated in claims (i) and (ii) of the theorem, are physically relevant, since they allow one to deduce - at least heuristically - features of the long time behavior of the Euler dynamics. See the next two remarks.
\end{remark}

\begin{remark}
\label{rem:4.2}
To better understand the connection between the toy model and the underlying Euler dynamics, we come back to the situation depicted at the beginning of the section. More precisely, we consider the reduced model of Subsection \ref{sec:2.1} with initial configuration,
\be
\label{p:9}
\omega_\eps(x,0) = \eps^{-2} \gamma(|x-z_*|/\eps)\;,
\ee
where $\gamma(r)$, $r\ge 0$, is a nonnegative smooth function with support $[0,1]$ such that $\int\!\rmd x\, \gamma(|x|) = 2\pi \int_0^1\! \rmd r\, \gamma(r) = 2\pi$. 

Consider the trajectory $x(x_0,t)$ of the fluid particle initially in $x_0$, i.e., the solution to \eqref{Cons4} with $u(x,t) = \int\!\rmd y\, K(x - y)\,\omega_\eps(y,t)$.  We next show that the toy model is consistent with \eqref{Cons4} if the vorticity remains radially distributed around the moving center $B(t)$. 

To this purpose, we assume that $\omega_\eps(x,t) = \eps^{-2} \gamma_t(|x-B(t)|/\eps)$, where $\gamma_t$ is a nonnegative smooth function with support in $[0,q_t]$ for some $q_t>0$, and $\int_0^1\! \rmd r\, \gamma_t(r) = 1$.  The corresponding velocity field can be easily computed, see, e.g., \cite[Eq.\ (2.14)]{BM}: letting $x-B(t) = r \mathbf{e}$ with $\mathbf{e} = (\cos\theta,\sin\theta)$, it is readily seen that $u(x,t) = \alpha (r)\, \mathbf{e}^\perp$ with
\begin{equation*}
\alpha(r) = -\frac{1}{\eps^2}\int_0^1\!\rmd\sigma\, \sigma\,\gamma_t(r\sigma/\eps) = \begin{cases} -r^{-1} & \text{if } r>\eps q_t\;, \\ -r^{-1}\int_0^{r/\eps}\!\rmd k\, k\, \gamma_t(k) & \text{if } r\le \eps q_t\;. \end{cases}
\end{equation*}

This shows that the field $u(x,t)$ produced by a symmetric blob is approximately equals to $g(x,B(t))$, see \eqref{g}, provided $x$ is far away from the center of vorticity. 

In absence of the external field $F$, the symmetry assumption is rigorously true at any time, since the vorticity distribution is stationary, $\omega_\eps(x,t) = \eps^{-2}\gamma(|x-z_*|/\eps)$, and each fluid particle performs a uniform circular motion around $z_*$. 

In presence of an external field $F$, each fluid particle departs, in general, from a uniform circular motion, thus destroying the initial symmetry of the vorticity distribution; this, in turn, produces a velocity $u(x,t)$ with a non-zero radial component, that enhances this effect.

On the other hand, assuming the vorticity ``frozen'' in a symmetric distribution, Theorem \ref{thm:6} establishes that each fluid particle remains very close to a circular motion for a very long time, due to an averaging mechanism that reduces the effect of the external field. Of course, this is only a toy model, because in the real case the vorticity distribution does not remain radially symmetric, and the problem becomes much more difficult. 

However, the analysis of the toy model suggests that this averaging mechanism could decrease the development of a non-symmetric component of the vorticity, thus preserving its concentration on long time scales. In Remark \ref{rem:4.3} below, we strengthen this conjecture by giving a more quantitative argument in the special case of the vortex patch dynamics. 
\end{remark}

\begin{remark}
\label{rem:4.3}
The name \textit{vortex patch dynamics} refers to the evolution of piecewise-constant vorticity configurations. It is well known that such configurations preserve their structure in time, and the problem of their time evolution is reduced to the so-called \textit{contour dynamics}, which governs the evolution of the boundaries of the regions with constant vorticity.

With this premise, we consider the reduced model of Subsection \ref{sec:2.1}, but now with initial configuration uniformly distributed in the disk $\Sigma(z_*|\eps)$, i.e.,
\be
\label{p:10}
\omega_\eps(x,0) = \begin{cases} 1/\eps^2 & \text{if } |x-z_*|\le\eps\,, \\ 0 & \text{otherwise}\,. \end{cases}
\ee

The evolved configuration $\omega_\eps(x,t)$ is a step function, equals to $1/\eps^2$ on its support $\Lambda_\eps(t)$. Clearly, $\Lambda_\eps(t) = \Sigma(z_*|\eps)$ in the absence of the external field $F$. Now, the study of the toy model in Theorem \ref{thm:6} suggests that the effect of the external field $F$ is such that $\Sigma_\eps^-(t) \subset \Lambda_\eps(t) \subset \Sigma_\eps^+(t)$ for any $t\in [0,\eps^{-\beta}]$, where $\Sigma_\eps^\pm(t) = \Sigma(B(t)|\eps \pm c\,\eps^{3-\beta})$ for a suitable $c>0$. 

Let us evaluate the velocity $u(x,t)$ of a fluid particle located at $x\in \Sigma_\eps^+(t)\setminus \Sigma_\eps^-(t)$. We decompose $u(x,t) = u_-(x,t) + u_+(x,t)$ with
\begin{equation*}
u_-(x,t) = \frac{1}{\eps^2} \int_{\Sigma_\eps^-(t)}\!\rmd y\, K(x-y)\;, \quad 
u_+(x,t) = \frac{1}{\eps^2} \int_{\Lambda_\eps(t) \setminus \Sigma_\eps^-(t)}\!\rmd y\, K(x-y)\;.
\end{equation*}

The component $u_-(x,t)$ is directed along $(x-B(t))^\perp$ by symmetry, and its intensity can be computed by arguing as in Remark \ref{rem:4.2}, getting
\begin{equation*}
u_-(x,t) = - \frac{1-c\,\eps^{2-\beta}}{\eps|x-B(t)|}\frac{(x-B(t)^\perp}{|x-B(t)|} = (1-c\,\eps^{2-\beta}) \, g(x,B(t))\;,
\end{equation*}
with $g$ as in in \eqref{g}.

The component $u_+(x,t)$, due to the fluid particles in the annulus $\Sigma_\eps^+(t)\setminus\Sigma_\eps^-(t)$, is the ``dangerous part'' of $u(x,t)$, since it has nonzero component along $x-B(t)$, but we claim that it is negligible as $\eps\to 0$. To this purpose, we decompose the annulus into the union of $2N+1$ disjoint annulus sectors $B_j$, $j=-N,\ldots, N$, of equal length $2\pi\eps/(2N+1) \simeq \eps^{3-\beta}$ and such that $x\in B_0$. We then estimate,
\begin{equation*}
|u_+(x,t)| \le \frac{1}{2\pi\eps^2} \int_{B_{-1}\cup B_0\cup B_1}\!\rmd y\, \frac{1}{|x-y|} + \sum_{|j|=2}^N \frac{1}{2\pi\eps^2} \int_{B_j}\!\rmd y\, \frac{1}{|x-y|}\;.
\end{equation*}
The first integral can be bounded by a rearrangement as done in \eqref{h2}, while the other integrals are easily controlled  noticing that $|y-x|\ge \const \eps^{3-\beta}|j|$ for $y\in B_j$ and $|j|=2,\ldots, N$. Therefore, denoting by $|B_j|$ the area of $B_j$, 
\begin{equation*}
\begin{split}
|u_+(x,t)| & \le \frac1{\eps^2} \sqrt{|B_{-1}\cup B_0\cup B_1|/\pi} + \const \frac1{\eps^2} \sum_{|j|=2}^N \frac{|B_j|}{\eps^{3-\beta}|j|} \\ & \le \const \big(\eps^{1-\beta} + \eps^{1-\beta}\log N \big) \le \const \eps^{1-\beta}|\log\eps| \;,
\end{split}
\end{equation*}
having used that $|B_j|\le \const \eps^{6-2\beta}$ and $N\le \const\eps^{-2+\beta}$.

In conclusion, we have shown that the difference $\delta u(x,t) = u(x,t)-g(x,B(t))$ is negligible in the limit $\eps\to 0$, validating the initial assumption of neglecting its effect with respect to that of the external force $F$. Clearly, this remains a heuristic argument, to make a rigorous proof we should get also a good control on the smoothness properties of $\delta u(x,t)$.
\end{remark}


\begin{thebibliography}{9}

\bibitem{Are79} \textsc{H. Aref}, \textit{Motion of three vortices},  Phys. Fluids, 22 (1979), pp.~393--400.

\bibitem{Are07} \textsc{H. Aref}, \textit{Point vortex dynamics: A classical mathematics playground}, Journ. Math. Phys., 48 (2007), 065401, 23 pp.

\bibitem{Bat} \textsc{G.K. Batchelor}, \textit{An introduction to fluid dynamics}, Cambridge University Press, Cambridge, 1967.

\bibitem{BCM00} \textsc{D. Benedetto, E. Caglioti, and C. Marchioro}, \textit{On the motion of a vortex ring with a sharply concentrate vorticity}, Math. Meth. Appl. Sci., 23 (2000), pp.~147--168.

\bibitem{BM} \textsc{A.L. Bertozzi and A.J. Majda}, \text{Vorticity and Incompressible Flow}, Cambridge Texts Appl. Math., Cambridge University Press, Cambridge, 2002.

\bibitem{CaM} \textsc{L. Caprini and C. Marchioro}, \textit{Concentrated Euler flows and point vortex model}, Rendiconti di Matematica e delle sue applicazioni, 36 (2015), pp.~11--25.

\bibitem{CGC} \textsc{W. Craig, C. Garcia-Azpeitia, and Chi-Ru Yang}, \textit{Standings waves in near-parallel vortex filaments}, Commun. Math. Phys., 350 (2017), pp.~175--203.

\bibitem{DuP} \textsc{D. Durr and M. Pulvirenti}, \textit{On the vortex flow in bounded domains}, Commun. Math. Phys., 85 (1983), pp.~265--273.

\bibitem{Gal} \textsc{Th. Gallay}, \textit{Interaction of vortices in weakly viscous planar flows}, Archive Rat. Mech. Anal., 200 (2011), pp.~445--490.

\bibitem{Hel} \textsc{H. Helmholtz}, \textit{On the integrals of the hydrodynamical equations which express vortex motion}, Phil. Mag., 33 (1867), pp.~485--512.

\bibitem{IfM} \textsc{D. Iftimie and C. Marchioro}, \textit{Self-similar point vortices and confinement of vorticity}, Comm. Partial Differential Equations (to appear).

\bibitem{ISG} \textsc{D. Iftimie, T. Sideris, and G. Gambin}, \textit{On the evolution of compactly supported planar vorticity}, Comm. Partial Differential Equations, 24 (1999), pp.~1709--1730.

\bibitem{Kel} \textsc{L. Kelvin}, \textit{Mathematical and physical papers}, Cambridge University Press, Cambridge, UK, 1910.

\bibitem{Kir} \textsc{G. Kirchhoff}, \textit{Vorlesungen Ueber Math. Phys.}, Teuber, Leipzig, 1876.

\bibitem{KMD} \textsc{R. Klein, A. Majda and K. K. Damodaran}, \textit{Simplified equations for the interaction of nearly parallel vortex filaments}, J. Fluid Mech., 288 (1995), pp.~201--248.

\bibitem{LoN} \textsc{M. Lopes Filho and H. Nussenzveig Lopes}, \textit{An extension of the C. Marchioro's bound on the growth of a vortex patch to flows with $L^p$ vorticity}, SIAM J. Math. Anal., 29 (1998), pp.~596--599.

\bibitem{MB} \textsc{A. J. Majda and A. L. Bertozzi}, \textit{Vorticity and incompressible flow}, Cambridge Texts in Applied Mathematics, Cambridge, U.K., 2002.

\bibitem{Mar88} \textsc{C. Marchioro}, \textit{Euler evolution for singular initial data and vortex theory: A global solution}, Commun. Math. Phys., 116 (1988), pp.~45--55.

\bibitem{Mar90} \textsc{C. Marchioro}, \textit{On the vanishing viscosity limit for two-dimensional Navier-Stokes equations with singular initial data}, Math. Meth. Appl. Sci., 12 (1990), pp.~463--470.

\bibitem{Mar94} \textsc{C. Marchioro}, \textit{Bounds on the growth of the support of a vortex patch}, Commun. Math. Phys., 164 (1994), pp.~507--524.

\bibitem{Mar96} \textsc{C. Marchioro}, \textit{On the growth of the vorticity support for an incompressible nonviscous fluid in a two-dimensional exterior domain}, Math. Meth. Appl. Sci., 19 (1996), pp.~53--62.

\bibitem{Mar98} \textsc{C. Marchioro}, \textit{On the localization of the vortices}, Boll. Unione Mat. Ital., 1-B(8) (1998), pp.~571--584.

\bibitem{Ma98} \textsc{C. Marchioro}, \textit{On the inviscid limit for a fluid with a concentrated vorticity}, Commun. Math. Phys., 196 (1998), pp.~53--65.

\bibitem{Mar99} \textsc{C. Marchioro}, \textit{Large smoke rings with concentrated vorticity}, Journ. Math. Phys., 40 (1999), pp.~869--883.

\bibitem{Mar07} \textsc{C. Marchioro}, \textit{Vanishing viscosity limit for an incompressible fluid with concentrated vorticity}, J. Math. Phys. 48 (2007), 065302, 16 pp.

\bibitem{MaN99} \textsc{C. Marchioro and P. Negrini}, \textit{On a dynamical system related to fluid mechanics}, NoDEA Nonlinear Diff. Eq. Appl., 6 (1999), pp.~473--499.

\bibitem{MaP} \textsc{C. Marchioro and E.Pagani}, \textit{Evolution of two concentrated vortices in a two-dimensional bounded domain}, Math. Meth. Appl. Sci., 8 (1986), pp.~328--344.

\bibitem{MaP83} \textsc{C. Marchioro and M. Pulvirenti}, \textit{Euler evolution for singular data and vortex theory}, Commun. Math. Phys., 91 (1983), pp.~563--572.

\bibitem{MaP84} \textsc{C. Marchioro and M. Pulvirenti}, \textit{Vortex methods in two-dimensional fluid dynamics}, Lecture Notes in Physics, vol.~203, Springer-Verlag, New York, 1984.

\bibitem{MaP93} \textsc{C. Marchioro and M. Pulvirenti}, \textit{Vortices and localization in Euler flows}, Commun. Math. Phys., 154 (1993), pp.~49--61.

\bibitem{MaP94} \textsc{C. Marchioro and M. Pulvirenti}, \textit{Mathematical theory of incompressible non-viscous fluids}, Applied mathematical sciences vol.~96, Springer-Verlag, New York, 1994.

\bibitem{NoS} \textsc{E. A. Novikov and Y. B. Sedov}, \textit{Vortex collapse}, Journal of Experimental and Theoretical Physics, 50 (1979), pp.~297--301.

\bibitem{ONe} \textsc{K. A. O'Neil}, \textit{Stationary configurations of point vortices}, Transactions of the American Mathematical Society, 302 (1987), pp.~383--425.

\bibitem{Poi} \textsc{H.Poincar\'e}, \textit{Theories des tourbillons}, George Carr\'e, Paris, 1893.

\bibitem{Ser} \textsc{P. Serfati}, \textit{Borne en temps de caract\'eristiques de l' \'equation d'Euler 2D \`a tourbillon positif et localisation pour le mod\'ele point-vortex}, preprint (1998).

\bibitem{Tur} \textsc{B. Turkington}, \textit{On the evolution of a concentrated vortex in an ideal fluid}, Arch. Ration. Mech. Anal., 97 (1987), pp.~75--87.

\end{thebibliography}
\end{document}